\documentclass[11pt]{article}
\usepackage{a4,epsfig,longtable,supertabular}
\usepackage{amsmath,amsthm,amsfonts}
\usepackage{graphicx}
\usepackage{fancyhdr}
\usepackage[kluwer]{harvard}
\usepackage{nicefrac}

\numberwithin{equation}{section}

\newtheorem{theorem}{Theorem}[section]
\newtheorem{proposition}[theorem]{Proposition}
\newtheorem{corollary}[theorem]{Corollary}
\newtheorem{lemma}[theorem]{Lemma}
\theoremstyle{definition}

\newtheorem{remark}[theorem]{Remark}

\newtheorem{example}[theorem]{Example}
\newtheorem{assumption}[theorem]{Assumption}

\newcommand{\R}{\mathbb{R}}

\newcommand{\N}{\mathbb{N}}
\newcommand{\F}{\mathcal{F}}
\newcommand{\G}{\mathcal{G}}

\newcommand{\vecb}[1]{{\boldsymbol#1}}

\newcommand{\norm}[1]{\left\| #1 \right\|}
\newcommand{\abs}[1]{\left |#1\right |}

\newcommand{\Ind}[1]{1_{#1}}

\renewcommand{\abs}[1]{\bigl|#1\bigr|}
\newcommand{\bbF}{\mathbb{F}}
\newcommand{\bbG}{\mathbb{G}}
\newcommand{\bbM}{\mathcal{M}}
\newcommand{\ind}[1]{1_{\{ #1\}}}

\newcommand{\prm}{\operatorname{pr}^{(m)}}

\begin{document}
\begin{center}
 \large { \Large \bf Corporate Security Prices  in Structural Credit Risk Models \\[0.2cm] with Incomplete Information: Extended Version\footnote{A shorter version of this paper will be published in Mathematical Finance.}}\\[0.8cm]
{\sc R\"udiger Frey\footnote{Corresponding author, Institute of Statistics and Mathematics,  Vienna University of Economics and Business, Welthandelsplatz 1,  A-1020 Vienna. Email: {\tt ruediger.frey@wu.ac.at}.} Lars R\"osler\footnote{Institute of Statistics and Mathematics, Vienna University of Economics and Business, {\tt lars.roesler@web.de}.}, Dan Lu\footnote{{\tt dan.lu@math.uni-leipzig.de} }}
  \\[0.5cm]
{\small \it Institute of Statistics and Mathematics, Vienna University of Economics and Business (WU) \\[0.2cm]
  \today}
\end{center}

\begin{abstract}
The paper studies derivative asset analysis in structural credit risk models where the asset value of the firm is not fully observable.  It is shown that in order to determine the price dynamics of traded  securities one needs to solve a stochastic filtering problem for the asset value. We transform this problem  to a filtering problem for a stopped diffusion process and we apply results from the filtering literature to this problem.  In this way we obtain an  SPDE-characterization  for the filter  density.  Moreover, we  characterize  the  default intensity under incomplete information and we determine the price dynamics  of traded securities. Armed with these results we study derivative asset analysis in our setup: we explain how the model can be applied to the pricing of options on traded assets and  we discuss dynamic hedging and model calibration.  The paper closes with  a  small simulation study.
\end{abstract}

\paragraph{Keywords.} Structural credit risk models, incomplete information,  stochastic filtering, derivative asset analysis for  corporate securities

\section{Introduction}
Structural credit risk models such as the first-passage-time models proposed by  \citeasnoun{bib:black-cox-76} or \citeasnoun{bib:leland-94} are widely used in the analysis of defaultable corporate securities.   In these models, a  firm defaults if a random process $V$ representing  the firm's  asset value  hits some threshold $K$  that is typically linked to value of the firm's liabilities. First-passage-time models offer an intuitive economic interpretation of the default event.
However, in the practical application of these models  a number of difficulties arise:
To begin with, it might be  difficult for investors in secondary markets to assess precisely the value of the firm's assets. Moreover, for tractability reasons $V$ is frequently  modelled as a diffusion process. In that case the default time $\tau$ is a predictable stopping time, which leads to unrealistically low values for short-term credit spreads.
For these reasons \citeasnoun{bib:duffie-lando-01} propose a model where secondary  markets  have only incomplete  information on the asset value  $V$. More precisely they consider the situation  where the market obtains at discrete time points $t_n$ a noisy accounting report of the form $Z_n = \ln V_{t_n} + \varepsilon_n$; moreover, the default history of the firm can be observed. Duffie and Lando show that in this setting the default time $\tau$ admits an intensity  that is proportional  to the derivative of the conditional density  of the asset value at the default threshold $K$. This well-known result provides an interesting link between structural and reduced-form models. Moreover, the result  shows that by introducing incomplete information it is possible to construct structural models where  short-term credit spreads take reasonable values.
The subsequent work of \citeasnoun{bib:frey-schmidt-09b}  discusses the pricing of the firm's equity in structural models with unobservable asset value.    Moreover, it is shown that the valuation of the firm's equity and debt leads to a stochastic  filtering problem: one needs to determine the conditional distribution of the current asset value $V_t$ given the $\sigma$-field $\F_t^{\bbM}$ representing the available information  at time $t$. \citeasnoun{bib:frey-schmidt-09b} consider this problem in the setup of  Duffie and  Lando where new information on the asset value arrives only  at discrete points in time. Working with a Markov-chain approximation approximation for $V$ they derive  a recursive updating rule for the conditional distribution of the approximating Markov chain  via  elementary Bayesian updating; the discrete nature of the information-arrival is crucial for their arguments.

Neither \citeasnoun{bib:duffie-lando-01} nor \citeasnoun{bib:frey-schmidt-09b}  study   the  price dynamics of traded securities under incomplete information.   Hence in these papers  it is not possible  to analyze  the pricing and the hedging of  derivative securities such as options on corporate bonds or on the stock.
The main  goal of the present paper  is therefore to develop  a proper  theory of derivative asset analysis for structural credit risk models under incomplete information.

More precisely, we make the following contributions. First, in order to obtain realistic price dynamics for the traded securities, we  model  the noisy observations of the asset value by a continuous time process of the form $Z_t = \int_0^t a(V_s) ds + W_t$ for some Brownian motion $W$ independent of $V$.
We show that this leads to price processes  with non-zero instantaneous volatility, whereas the discrete information arrival considered by Duffie and Lando or Frey and Schmidt generates asset prices  that evolve deterministically between the news-arrival dates. Moreover, modeling $Z$ as a continuous time processes is in line with the standard literature on stochastic filtering such as  \citeasnoun{bib:bain-crisan-08}.
Second,  in order to derive the price dynamics of traded securities we determine the dynamics of the conditional distribution of $V_t$ given  $\F_t^\bbM$.
This is a challenging  stochastic filtering problem, since under full observation the default time $\tau$  is predictable, so that   standard filtering techniques for point process observations  (see for instance \citeasnoun{bib:bremaud-81})  do not apply.  We therefore transform   the original problem to a new filtering problem where the   observations consist only of the process $Z$;  the signal process in this new problem  is on the other hand given by  the asset value process stopped at the first exit time of the solvency region $(K, \infty)$. Using results of \citeasnoun{bib:pardoux-78} on the filtering of stopped diffusion processes we derive a stochastic partial differential equation (SPDE) for the conditional density   of $V_t$ given $\F_t^\bbM$, denoted $\pi(t,\cdot)$, and we discuss the numerical solution of this SPDE via a Galerkin approximation.  Extending the work of \citeasnoun{bib:duffie-lando-01} to our more general information structure,  we show that $\tau$ admits an intensity process $(\lambda_t)_{t \ge 0}$ such that the intensity at time $t$   is proportional  to the spatial derivative of  $\pi(t,v)$ at $v=K$.
Armed with these results we  finally study derivative asset analysis in our setup:   we identify the price dynamics of the traded securities;  we consider  the pricing of options on traded assets;  we  derive risk-minimizing dynamic hedging strategies for these claims,  and we discuss model calibration. The  paper closes with a small simulation study  illustrating  the  theoretical results.

Incomplete information and filtering methods have been used before in the analysis of credit risk. Structural models with incomplete information were considered among others by \citeasnoun{bib:kusuoka-99}, \citeasnoun{bib:duffie-lando-01},  \citeasnoun{bib:jarrow-protter-04},
\citeasnoun{bib:coculescu-geman-jeanblanc-06}, \citeasnoun{bib:frey-schmidt-09b} and \citeasnoun{bib:cetin-12}. The last contribution is related to the present paper.   Working in a similar setup as ours, Cetin uses  probabilistic arguments   to establish  the existence of a default intensity  with respect to $\bbF^\bbM$, and he derives the corresponding filter equations. He does not discuss the existence of the conditional density  $\pi(t,\cdot)$ and  financial applications are discussed only in a peripheral manner. From a purely mathematical point of view our paper is also closely related to  \citeasnoun{bib:krylov-wang-11} who deal with the filtering of partially observed diffusions up to the first exit time of a domain. The relation between our results and those of Krylov and Wang are best explained once the mathematical details of our setup have been introduced, and we refer to Section~\ref{sec:filtering}, Remark~\ref{rem:krylov} for a deeper discussion of similarities and differences between the two papers.

Reduced-form credit risk models with incomplete information have been considered
previously by \citeasnoun{bib:duffie-et-al-06}, \citeasnoun{bib:frey-runggaldier-10} and  \citeasnoun{bib:frey-schmidt-12}, among others.  The  modelling philosophy  of the present paper is inspired by  \citeasnoun{bib:frey-schmidt-12}, but the mathematical analysis  differs substantially. In particular, in \citeasnoun{bib:frey-schmidt-12} the default times of the firms under consideration do admit an intensity under full information. Hence the filtering problem that arises in the pricing of credit derivatives can be addressed  via a straightforward application of the innovations approach to nonlinear filtering.

The remainder of the paper is organized as follows. In Section~\ref{sec:model} we introduce the model; the relation between traded  securities and stochastic filtering  is discussed in Section~\ref{sec:pricing-basic-securities}; Section~\ref{sec:filtering} is concerned with the stochastic filtering of the asset value; in Section~\ref{sec:dynamics} we derive the dynamics of corporate securities; Section~\ref{sec:applications} is concerned with  derivative asset analysis; the results of numerical experiments are given in Section~\ref{sec:simulations}.

\paragraph{Acknowledgements.} Financial support from the German Science Foundation (DFG)
and from the Vienna Science and Technology Fund (WWTF), project MA14-031, is gratefully acknowledged. Moreover, we thank Andreas Peterseil for very competent  research assistance and several anonymous referees for useful comments.

\section{The Model}
\label{sec:model}

We begin by introducing the mathematical structure of the model.
We work on  a filtered probability  space $(\Omega, \mathcal{G}, \mathbb{G}=
(\mathcal{G}_t)_{t\geq 0}, Q)$ and  we assume that  all processes introduced below are
$\mathbb{G}$-adapted. Since we are mainly interested in the pricing of derivative securities we  assume that $Q$ is the risk-neutral pricing  measure.
We consider a company with nonnegative asset value process $V=(V_t)_{t\geq 0}$. The company is subject to default risk and the default time is modelled as a first passage time, that is
\begin{equation} \label{eq:def-tau}
\tau=\inf\{t \geq 0 \colon V_t\leq K\}
\end{equation}
for some default threshold  $K>0$. In practice  $K$ might represent  solvency capital requirements imposed by regulators (see Example~\ref{ex:fin-institution}) or it might correspond to an endogenous default threshold as in \citeasnoun{bib:duffie-lando-01} (see Example~\ref{ex:duffie-lando}).
By $Y_t=1_{\{\tau\leq t\}}$ we denote the \emph{default state} of the firm at time $t$, that is  $Y_t = 1$ if and only if the firm has defaulted by time $t$; the associated \emph{default indicator process} is denoted by $Y = (Y_t)_{t \ge 0}$.

\begin{assumption}[Dividends and asset value process] \label{ass:assets-and-dividends}
1) The risk free rate of interest is constant and equal to $r \ge 0$.

2) The firm pays \emph{dividends}   at equidistant  deterministic time points $t_1$, $t_2, \ldots$ (for instance semi-annual dividend payments). The set of dividend dates is denoted by $\mathcal{T}^D$. The  dividend payment at $t_n$ is a random percentage of the \emph{surplus} $(V_{t_n-}-K)^+$ (the part of the asset value that can be distributed to shareholders  without sending  the company into immediate default).  Denoting by  $d_n$  the   dividend payment at $t_n$, it  holds that
\begin{equation}\label{eq:dividend-size}
    d_n=\delta_n (V_{t_n} -K)^+;
\end{equation}
here  $(\delta_n)_{n=1,2,\dots}$ is  an iid sequence of noise variable that are independent of $V$,
take  values in $(0,1)$ and that have   density function $\varphi_\delta $. We assume that $\varphi_\delta$ is bounded and twice continuously differentiable on $[0,1]$ with $\varphi_\delta(1)=0$.   For $V_{t_n -} > K$ the conditional distribution of $d_n$ given the history of the asset value process is thus of the form $\varphi(y , V_{t_n-}) dy$  where
\begin{equation}\label{eq:dividend-density}
  \varphi(y,v) = \frac{1}{(v-K)} \varphi_\delta\Big( \frac{y}{(v-K)}\Big)\ind{v >K}\,.
\end{equation}
Let  $D_t=\sum_{\{n\colon t_n\leq t\}} d_n$  so that $D = (D_t)_{t\ge 0}$ is the cumulative dividend process. In the sequel we  denote by
$\mu^D(dy,dt)$ the random measure associated with the sequence $(t_n, d_n)_{n \in \mathbb{N}}$.

3)  The \emph{asset value process} $V = (V_t)_{t \ge 0} $ has the following dynamics
\begin{equation}\label{eq:dVt}
V_t = V_0 + \int_0^t r V_s ds + \int_0^t \sigma V_s dB_s - \kappa  D_t
\end{equation}
for a constant volatility $\sigma>0$,  a standard $Q$-Brownian motion $B$ and a random variable $V_0$. The parameter $\kappa $ takes values in $\{0,1\}$. For $\kappa =1$ (the most relevant case) the asset value is reduced at a dividend date by the amount $d_n$ distributed to shareholders;   $\kappa =0$ corresponds to the case where we view the $d_n$ merely as noisy signal of the asset value and not as a payment to shareholders (see Example~\ref{ex:duffie-lando}).
We assume that $V_0$ has Lebesgue density $\pi_0(v)$ for a continuously differentiable function $\pi_0 \colon [K,\infty) \to \R^+$ with $\pi_0(K) =0$ such that $V_0$ has finite second moment.
\end{assumption}

The second assumption  reflects the fact that in reality there is a positive but noisy relation between asset value and dividend size. Note that it follows from \eqref{eq:dividend-density}  that $d_n < (V_{t_n -} - K )^+$. This restriction on the dividend size can be viewed as implicit protection of debtholders as it ensures that the firm will not default at a dividend date due to an overly large dividend. Together with our assumptions on $\varphi_\delta$,   \eqref{eq:dividend-density} implies that for a given $d >0$, $\varphi(d,v)$ is zero for all $v$ such that $d/(v-K) \ge 1$, that is for $v \le d+K$. Moreover, it holds that
\begin{equation} \label{eq:bound-on-density}
\sup_{v \ge K} \varphi (d,v) \le \frac{1}{d} \max_{\delta \in [0,1]} \varphi_\delta (\delta)\,.
\end{equation}
Note that the dividend policy \eqref{eq:dividend-size} is not the outcome of a formal optimization process. In fact, as shown for instance in \citeasnoun{bib:jeanblanc-shiriayev-95}, it might be optimal to pay out a larger fraction of the available surplus if $V_{t_n}$ is large. While the filtering results in Section~\ref{subsec:filtering-wrt-FI} could  be extended to such a setup, provided the conditional density $\varphi (\cdot,v)$ of the dividend size satisfies certain regularity conditions, the pricing of the firm's stock would become more involved. Moreover, dividend policies adopted in practice are  guided to a large extent  by market conventions and rules of thumb. For these reasons we stick to the simple rule \eqref{eq:dividend-size}.

The  $\mathbb{G}$-compensator of the random measure $\mu^D$ associated with the sequence $(t_n, d_n)_{n \in \mathbb{N}}$ is  given by $\gamma^D(dy,dt) =  \sum_{n=1}^\infty \varphi(y,V_{t_n -}) dy \,\delta_{\{t_n\}}(dt)$. Note that for $g \colon [0, \infty) \times [0, \infty) \to \R^+$ it holds that
$$ \int_0^\infty \int_0^\infty g(t,y) \gamma^D(dy,dt) = \sum_{n=1}^\infty \int_0^\infty  g(t_n, y) \varphi(y,V_{t_n -}) dy \,.$$

The assumption that between dividends the asset value  is a geometric Brownian motion is routinely made in
the literature on structural credit risk models such as \citeasnoun{bib:leland-94}
or \citeasnoun{bib:duffie-lando-01}. For empirical support for the assumption of geometric Brownian motion as a model for the asset price dynamics we refer to~\citeasnoun{bib:moodys-edf-12}. Note that the assumption that $V$ follows a geometric Brownian  motion does not imply that the  stock price follows a geometric Brownian motion. In fact, our analysis in Section~\ref{sec:dynamics} shows that in our setup  the stock price dynamics can be much `wilder' than geometric Brownian motion.
Note finally that  $V$ is not a traded asset so that its drift under $Q$ might in principle be different from the risk-free rate $r$. However, setting the drift of $V$ equal to $r$ permits us to interpret $V$ as value of all future dividend payments of the firm (up to $t=\infty$), see Lemma~\ref{lemma:finite-stock-price} below.

\paragraph{Information structure and  pricing.} In our setting the asset value $V$ is not directly observable. Instead we assume that prices of corporate securities are determined as conditional expectation with respect to some filtration $\mathbb{F}^\bbM = (\mathcal{F}_t^\bbM)_{t\geq 0}$ that is generated by the default history, by  the dividend payments of the firm and by   observations of functions of $V$ in additive Gaussian noise:
\begin{assumption} \label{ass:information}  It holds that $\mathbb{F}^\bbM =\mathbb{F}^Y \vee \mathbb{F}^D\vee\mathbb{F}^Z$, where  $\mathbb{F}^Y$ denotes the filtration generated by the
default indicator process $Y$, where $\mathbb{F}^D$ denotes the filtration generated by   $D$ and where the filtration $\mathbb{F}^Z$ is generated by the $l$-dimensional process $Z$  with
\begin{equation}
{Z}_t=\int_0^t{a}(V_s)ds+ {W}_t\,.
\end{equation}
Here  ${W}$  is an $l$-dimensional    $\bbG$-Brownian motion independent of $B$,
and   $a = (a^1, \dots, a^l)$ is a bounded and  continuously differentiable  function from $\mathbb{R^+}$ to
$\mathbb{R}^l$ with $a(K) =0$. Note that the assumption $a(K) =0$ is no real restriction  as  the function $a $ can be replaced
with $a - a(K)$ without altering the information content of $\bbF^\bbM$.
\end{assumption}

In the sequel $\mathbb{F}^\bbM$ is called  \emph{modeling filtration}, since it represents the fictitious flow of information that is employed in  the \emph{construction} of the  model. In particular, we will not associate the process  $Z$ with publicly observable economic data; it is simply a mathematical device  that generates  the diffusive component in the asset price dynamics.  As explained in the next section,  for the \emph{application }  of the model, that is for  pricing  and hedging  of derivative securities, it is sufficient  to observe the price processes of  traded assets and the default history of the firm. This is important since pricing formulas and hedging strategies need to be computed in terms of publicly available information.

We use martingale modelling to construct the price processes of traded securities and  we define  the   ex-dividend price of a generic traded security with  $\bbF^{\bbM}$-adapted cash flow  stream $(H_t)_{0 \le t \le T}$ and maturity date $T \in (0, \infty]$ by
\begin{equation} \label{eq:def-price-of-H}
\Pi_t^H = E^Q \Big ( \int_t^T e^{-r(s-t)} d H_s \mid \F_t^\bbM \Big )\, , \quad t \le T ,
\end{equation}
provided of course that the discounted cash flow stream is $Q$ integrable.  The use of the risk-neutral pricing formula~\eqref{eq:def-price-of-H} ensures that  the discounted gains from trade of every traded security are martingales, which is sufficient to exclude arbitrage opportunities. Hedging arguments within the  martingale modelling paradigm are presented in Section~\ref{subsec:hedging}.

\vspace{0.1cm}

Finally  we describe two economic settings that can be  embedded in our framework.

\begin{example} \label{ex:fin-institution} Our first example is that of  a financial institution that is subject to financial regulation. We assume that the institution has  issued shares to outside shareholders. It is run by a management team that knows  the asset value $V$.  Management is prevented from actively trading the shares of the institution, for instance because of insider trading regulation.   Outside stock and bond investors on the other hand are unable to discern the exact asset value from public information. The dividend policy of the firm is of the form \eqref{eq:dividend-size}.   In this example we let  $\kappa =1$ so that dividend payments do reduce the asset value of the firm.
We assume that the institution  is subject to capital adequacy rules such as the Basel~III or the  Solvency~II rules.   Loosely speaking these rules require that the ratio of the equity capital of the firm over its total asset value must be larger than a given threshold $\gamma \in (0,1)$. If we denote by $\tilde K$ the value of the firms liabilities, this translates into the condition  that $(V_t - \tilde K) / V_t > \gamma$ and hence that
\begin{equation} \label{eq:capital-adequacy}
V_t > K : = \tilde K/(1-\gamma).
\end{equation}
We assume that regulators actively monitor that the state of the firm is in accordance with the capital adequacy rule~\eqref{eq:capital-adequacy} and that  management provides them with  correct information about the asset value.  If $V$ falls below  $K$, regulators shut down  the financial institution and there is  a default. Hence the default time is a first passage time with default threshold $K$ given in \eqref{eq:capital-adequacy}.
Note that in this setting default is enforced by regulators with privileged access to information.
\end{example}

\begin{example} \label{ex:duffie-lando} The well-known model of  \citeasnoun{bib:duffie-lando-01} can be embedded in our setup as well.
Duffie and Lando consider  a firm that is operated by risk-neutral equity owners who have complete information about $V$. The firm issues some debt in the form of a consol bond in order to profit from the tax shield of debt, but there are no traded shares and no dividend payments to outside investors. Equity owners are prohibited from trading in bond markets by insider trading regulation. In this setup the owners of the firm  have the option to stop servicing the firm's debt, in which case the firm defaults.  Following \citeasnoun{bib:leland-toft-96}, Duffie and Lando show that the optimal default time (for the equity owners)  is  a  first passage time,  but now with endogenously determined  default threshold $K$.

Note that in this example the random variables $d_n$ can be viewed  as additional  information on $V$ that arrives at discrete time points, such as earnings announcements. This interpretation  corresponds to a value of $\kappa =0$ for the parameter $\kappa$ in \eqref{eq:dVt}. Moreover,  the rvs $d_n$ do not have to be  of the special form \eqref{eq:dividend-size}; it suffices that for fixed $d$ the mapping $ v \mapsto \varphi(d,v)$ is smooth and bounded.
\end{example}

\section{Prices of Traded  Securities and Stochastic Filtering}
\label{sec:pricing-basic-securities}

In this section we explain the relation between the prices of traded securities and stochastic  filtering and we discuss several examples.

\paragraph{Traded securities.}  The set of traded securities consists so-called \emph{basic debt securities} and   of the  stock of the firm.   We now describe  the payoff stream of these  securities  in more detail.
First we refer to an asset as a \emph{basic debt security}  if its   cash-flow stream can be expressed as a linear combination of the following two building blocks
\begin{itemize}
\item[i)] A \emph{survival claim} with generic maturity date $T$. This claim  pays one unit of account at $T$, provided that $\tau >T$.
\item[ii)] A \emph{payment-at-default claim} with generic maturity date $T$. This claim pays one unit directly at $\tau$, provided that $\tau \le T$.
\end{itemize}
It is well known that   bonds issued by the firm  and  credit default swaps  on the firm can be expressed as linear combination of these building blocks,  see for instance  \citeasnoun{bib:Lando-98}.

Next we  discuss the modelling  of the firm's stock.  The shareholders of the firm receive the dividend payments made by the firm at dividend dates $t_n < \tau$. Hence the cumulative cashflow stream received by the shareholders  up to time $t$ equals  $H_t^\text{stock} = D_{t \wedge \tau}$. The risk neutral pricing formula~\eqref{eq:def-price-of-H}  thus implies that the value
of the  firm's stock\footnote{Note that strictly speaking  $S_t$ gives the market capitalization of the firm at time $t$, that is  value of the entire outstanding stock. Since we assume that the number of outstanding shares is constant we use the symbol $S$ also for the price process of a single share.} is given by
\begin{equation} \label{eq:def-stock-price}
S_t = E^Q \Big ( \sum_{\{n \colon t_n >t\}}  1_{\{\tau >t_n \}} e^{-(t_n-t)} d_n  \mid \F_t^\bbM \Big )\,.
\end{equation}
Note that the cash-flow stream of a basic debt security and of the stock is adapted to $\bbF^Y \vee \bbF^D$ and hence also to the modelling filtration $\bbF^\bbM$.

\paragraph{Relation to stochastic filtering.} Consider now a traded security with cash-flow stream $(H_t)_{0 \le t \le T}$ and ex-dividend price $\Pi_t^H = E^Q \big ( \int_t^T e^{-r(s-t)} d H_s \mid \F_t^\bbM \big )$.
In the sequel we mostly consider the pre-default value of the  security given by $\ind{\tau
>t} \Pi_t^H$ (pricing for $\tau \le t $ is largely related to the modelling of
recovery rates which is of no concern to us here). Using iterated conditional expectations we get that
\begin{equation} \label{eq:def-price-of-H-2}
\ind{\tau >t} \Pi_t^H  =
 E^Q \Big ( E^Q \big ( \ind{\tau >t} \int_t^T e^{-r(s-t)} d H_s \mid \G_t \big ) \mid \F_t^\bbM \Big ).
\end{equation}
By the Markov property of $V$, for basic debt securities and for the stock the inner conditional
expectation can be expressed as a function of time and of the current asset value $V_t$, that is
\begin{equation}\label{eq:def-full-info-value}
E^Q \Big (1_{\{\tau > t\}} \int_t^T e^{-r(s-t)} d H_s \mid \G_t \Big ) =
 1_{\{\tau > t\}}h(t,V_t)\,.
\end{equation}
The function   $h$ is  called the \emph{full-information value} of the security.
In Section~\ref{sec:filtering}  we  show that on the  $\F_t^\bbM$-measurable set $\{\tau >t\}$ the conditional distribution of $V_t$ given $\F_t^\bbM$ admits a density $\pi(t,\cdot) \colon [K, \infty) \to \R^+$  and we derive an SPDE for this density.
Substituting \eqref{eq:def-full-info-value} into \eqref{eq:def-price-of-H-2}  gives that
\begin{equation} \label{eq:pricing-via-filtering}
1_{\{\tau > t\}} \Pi_t^H = 1_{\{\tau > t\}}  E^Q \big (h(t,V_t) \mid \F_t^\bbM) = 1_{\{\tau > t\}} \int_K^\infty h(t,v)\pi(t,v) d v\,.
\end{equation}
Relation \eqref{eq:pricing-via-filtering} provides an important  relationship between prices of traded securities and stochastic filtering, which is used in two ways.   First, at any given time point $t_0$  an estimate of $\pi(t_0)$ is backed out  from the  price of traded  securities at time $t_0$, so that $\pi(t_0) $ can be viewed as a function of  observed prices  (the necessary  calibration methodology is described in Section~\ref{subsec:calibration}). Moreover, in Section~\ref{subsec:derivatives} we  show that  the price at time $t_0$  of an  option on the traded assets is a function of $\pi(t_0)$. Hence option prices can be evaluated using observable quantities (prices of traded securities) as input.
Second, in order to derive the price {dynamics} of  traded securities under the risk-neutral measure $Q$  we  determine the dynamics of  $\pi(t)$ using filtering methods; using \eqref{eq:pricing-via-filtering}  this gives  the dynamics of the pre-default value $1_{\{\tau > t\}} \Pi_t^H$ of the traded securities.

This approach is akin to the use of factor models in term structure modelling where prices of traded securities  are used to estimate the current value of the factor process  and where  bond price dynamics are derived from the dynamics of the factor process. In fact,  our model can be viewed as factor model with infinite-dimensional factor process $\pi(t)$.

\begin{remark} Our modelling strategy leads to filtering problems under $Q$ and  differs from the `classical' application of stochastic filtering  in statistical inference.
A typical problem in the latter  context would be as follows:  the process $Z$ is identified with  a specific set of economic data that contains noisy information of $V$, and filtering techniques are employed to estimate the conditional distribution of $V_t$ under the historical measure $P$  given  the observed trajectories of  $Z$, $D$ and $Y$ up to time $t$. Such an approach could be used  to estimate the firm's real-world default probability, similar in spirit to the well-known public firm EDF model~\citeasnoun{bib:moodys-edf-12}. It is worth mentioning that the mathematical results developed in Sections~\ref{sec:filtering}~and~\ref{sec:dynamics} cover also applications of this type.
\end{remark}

\paragraph{Full information value of traded securities.} Next  we discuss the computation of the  full information value  $h$ for basic debt securities and for the stock.  We concentrate on the case $\kappa=1$, so that there is a downward jump in $V$  at the dividend dates; for $\kappa =0$ the asset value is a geometric Brownian motion and the ensuing computations are fairly standard.

We begin with  a  survival claim with payoff $H_T = 1_{\{\tau > T\}}$ and associated full-information value   $h^{\text{surv}}(t,V_t)$.
Since $e^{-r t} h^{\text{surv}}(t,V_t)$ is  a $\mathbb{G}$-martingale
we get the following PDE characterization of $h^{\text{surv}}$: first, between dividend dates $h^{\text{surv}}$ solves the boundary value problem
$ \frac{d }{dt} h^{\text{surv}} +  \mathcal{L} h^{\text{surv}} = r h^{\text{surv}}$   with  boundary condition $h^{\text{surv}}(t, K) = 0$ for  $0 \le t \le T$,
 where we let for $f\in \mathcal{C}^2(0,\infty)$
\begin{equation}\label{eq:def-L}
\mathcal{ L} f (v) = r v \frac{df(v)}{dv}  + \frac{1}{2}\sigma^2 v^2 \frac{d^2 f(v)}{dv^2}\,;
\end{equation}
 second, at a dividend date $t_n \le T$ it holds that
\begin{equation} \label{eq:full-info-value-at-tn-surv}
h^{\text{surv}} (t_{n}-,v) = \int_0^{v-K} h^{\text{surv}}(t_n, v- y) \varphi(y,v) dy\,;
\end{equation}
finally, one has the terminal condition  $h^{\text{surv}}(T, v) = 1 $ for $v >K$. These conditions  can be used to compute $h^{\text{surv}}$ numerically by a backward induction over the dividend dates; see for instance \citeasnoun{bib:vellekoop-nieuwenhuis-06} for details.  Moreover, we will need the PDE characterization of  $h^{\text{surv}}$ to derive the price dynamics of a survival claim under incomplete information in Section~\ref{subsec:filtering-dynamics}.
Recall that a payment-at-default claim with maturity $T$ pays one unit
directly at $\tau$, provided that $\tau \le T$. The PDE characterization is similar to the case of a survival claim; however, now the boundary condition is
$h^{\text{def}}(t, K) = 1$, $0 \le t \le T$ and the terminal value is $h^{\text{def}}(T, v) = 0 $, $v >K$. By definition the full information value of all basic debt securities can be computed from $h^{\text{surv}}$ and $h^{\text{def}}$.

Next we consider the stock of the firm. It follows from \eqref{eq:def-stock-price} that the full information value of the firm's stock is given by
\begin{equation} \label{eq:def-S-full-info}
h^{\text{stock}}(t, v) = E^Q \Big ( \sum_{n \colon t_n >t}  1_{\{\tau >t_n \}} e^{-r (t_{n}-t) } d_n \mid V_t =v \Big )
\end{equation}
The next lemma whose proof is given in  Appendix~B shows that  $V_t$ can be interpreted as value of all future dividend payments (up to $T= \infty$).
\begin{lemma} \label{lemma:finite-stock-price}
Under Assumption~\ref{ass:assets-and-dividends} it holds that $E^Q \big ( \sum_{t_n \ge t}   e^{-r (t_{n}-t)} d_n \mid V_t \big) = V_t $.
\end{lemma}
Note that Lemma~\ref{lemma:finite-stock-price} implies that $h^{\text{stock}}(t, v) < v$.  It follows that the stock price $S_t$ is finite as well (this is not a priori clear since $S_t$ is the expected value of an infinite payment stream).
Using the fact that $e^{-rt} h^{\text{stock}}(t,V_t) + \int_0^t 1_{\{\tau >s\}} e^{-rs} dD_s$ is a $Q$ martingale we obtain  the following PDE characterization for $h^{\text{stock}}$: between dividend dates  $h^{\text{stock}}$ solves the PDE $ \frac{d }{dt} h^{\text{stock}} +  \mathcal{L} h^{\text{stock}} = r h^{\text{stock}}$   with  boundary condition $h^{\text{stock}}(t, K) = 0$; at the dividend date $t_n$  $h^{\text{stock}}$ satisfies the relation
\begin{equation} \label{eq:full-info-value-at-tn}
h^{\text{stock}} (t_{n}-,v) = \int_0^{v-K} \left( h^{\text{stock}}(t_n, v - y) + y \right) \varphi(y,v) dy\,.
\end{equation}
Since we assumed  equidistant dividend dates it holds that $h^{\text{stock}} (t,v) = h^{\text{stock}} (t + \Delta t,v)$ for $\Delta t = t_n - t_{n-1}$ so that it is enough to compute $h^{\text{stock}} (t,v) $ for $0 \le t \le t_1$.
An explicit formula for $h^{\text{stock}}$ is not available; the main problem is the fact that the  downward jump in $V$ at a dividend date combines  arithmetic and geometric expressions. There are essentially two options for computing $h^{\text{stock}}$ numerically. On the one hand one can rely on Monte Carlo methods. In order to speed up the simulation explicit pricing formulas for $h^{\text{stock}}$ in a Black Scholes model with  continuous dividend stream can be used  as control variate.
Alternatively, it is  possible to use PDE methods in order to compute  $h^{\text{stock}}$. We omit the details since the numerical computation of option prices is not central to our analysis.


\section{Stochastic Filtering of the Asset Value} \label{sec:filtering}

Fix some horizon date $T$, for instance the  largest maturity date of all outstanding  derivative securities related to the firm. Recall from the previous section that in order to derive the price {dynamics} of  traded securities  we  need to determine the dynamics of the conditional density  $\pi(t)$, $0 \le t <  \tau \wedge  T $.  This problem is studied in the present section.


\subsection{Preliminaries} \label{subsec:preliminaries}

Following the usual approach in stochastic filtering we start with a characterization of the conditional distribution of $V_t$ given $\F_t^\bbM$ (the filter distribution) in weak form. More precisely, given a function  $h $ on $[K, \infty)$ such that $E\big ( |h(V_t)|\big )< \infty $ for all $t \le T$ we want to derive the dynamics of the conditional expectation
\begin{equation} \label{eq:basic-filtering-problem}
1_{\{\tau > t\}}  E^Q \big ( h(V_t) \mid \F_t^\bbM \big), \quad t \le T\,.
\end{equation}
This is sufficient for our purposes since we are only interested in the dynamics of the filter distribution prior to default.

The problem \eqref{eq:basic-filtering-problem} is a challenging filtering problem  since the default time $\tau$  does not admit an intensity under full information. Hence  standard
filtering techniques for point process observations as in \citeasnoun{bib:bremaud-81} do
not apply. This issue is addressed in the following proposition where, loosely speaking, \eqref{eq:basic-filtering-problem} is transformed to a filtering problem with respect to the background filtration $\bbF^Z \vee \bbF^D$.
\begin{proposition} \label{prop:reduction-to-background}
Denote by $V^\tau =  (V_{t \wedge\tau})_{t \ge 0}$ the asset value process stopped at the
default boundary, by $\tilde{Z}_t = \int_0^t a( V_s^\tau) ds + W_t$ the observation of $V^\tau$ in additive Gaussian noise and by $\tilde D_t=\sum_{\{n\colon
t_n\leq t\}} \delta_n (V_{t_n -}^\tau -K)^+ $ the cumulative dividend process corresponding to
$V^\tau$. Then we have for  $h:[K ,\infty) \to \R$  such that $E^Q\big(| h(V_t)| \big) < \infty$ for all $t \le T$
\begin{equation}\label{eq:reduction-to-background}
\ind{\tau >t}  E^Q ( h(V_t) \mid \F_t^\bbM ) =
\ind{\tau >t} \frac{E^Q \big (
h(V_t^\tau) \ind{V_t^\tau > K} \mid \F_t^{\tilde Z} \vee\F_t^{\tilde{D}} \big ) }{ Q\big ( V_t^\tau > K \mid \F_t^{\tilde{Z}} \vee \F_t^{\tilde{D}} \big )}\,.
\end{equation}
\end{proposition}
\begin{proof} For notational simplicity we ignore the dividend observation in the proof
so that $\bbF^{\bbM} = \bbF^{Z} \vee \bbF^Y$. The first step is to show that
\begin{equation} \label{eq:different-filtration-1}
 E^Q \big ( h(V_t) \ind{\tau >t} \mid \F_t^\bbM \big ) =
 E^Q \big ( h(V_t^\tau) \ind{\tau >t} \mid \F_t^{\tilde Z } \vee \F_t^Y \big )\,,
\end{equation}
where the filtration $\bbF^{\tilde Z}$ is generated by  the noisy observations of the \emph{stopped} asset value process;
the proof of this identity is given in  Appendix~B.

Second, using the Dellacherie formula (see for instance Lemma~3.1 in
\citeasnoun{bib:elliott-jeanblanc-yor-00}) and the relation $\{ \tau > t\} = \{V_t^\tau
> K \}$, we  get
\begin{align*}
 E^Q \big ( h(V_t^\tau) \ind{\tau >t} \mid \F_t^{\tilde Z } \vee \F_t^Y \big ) &=
 \ind{\tau >t} \frac{E^Q \big ( h(V_t^\tau ) \ind{\tau >t} \mid \F_t^{\tilde Z}  \big ) }{
 Q\big (\tau >t \mid \F_t^{\tilde Z} \big )}\\
 &=  \ind{\tau >t} \frac{E^Q \big ( h(V_t^\tau)\ind{V_t^\tau > K} \mid \F_t^{\tilde Z}  \big ) }{
 Q\big (V_t^\tau > K \mid \F_t^{\tilde Z} \big )} \,,
\end{align*}
as claimed.
\end{proof}
With the notation $f(v) := h(v) \ind{v >K}$, Proposition~\ref{prop:reduction-to-background} shows that in order to evaluate the right side of \eqref{eq:reduction-to-background} one needs to compute for generic $f :[K, \infty) \to \R$  such that $E^Q \big( |f(V_t^\tau)| \big) < \infty$ for all $0 \le t \le T$    conditional expectations of the
form
\begin{equation} \label{eq:new-filtering-problem-1}
E^Q \big ( f(V_t^\tau)  \mid \F_t^{\tilde Z} \vee \F_t^{\tilde D} \big )\,.
\end{equation}
This is a stochastic filtering problem with signal process given by $V^\tau$ (the asset value process stopped at the first exit time of the halfspace  $(K, \infty)$).
In the sequel we  study this  problem  using  results of \citeasnoun{bib:pardoux-78} on the  filtering of diffusions stopped at the first exit time of  some \emph{bounded} domain, first for the case without dividends and in Section~\ref{subsec:filtering-wrt-FI} for the general case.
In order to apply  the results of \citeasnoun{bib:pardoux-78}    we fix some large number $N$ and replace the unbounded  halfspace $(K, \infty)$ with the  bounded domain $(K,N)$. For this we define the stopping time  $\sigma_N = \inf \{ t \ge 0 \colon V_t \ge N\}$ and we replace the original asset value process $V$ with the stopped process $V^N := (V_{t \wedge \sigma_N})_{t\ge 0}$.  Applying
Proposition~\ref{prop:reduction-to-background} to the  process $V^{N}$ leads to
a filtering problem with signal process $X := (V^N)^\tau$. More precisely, one has to compute  conditional expectations of the form
\begin{equation}\label{eq:new-filtering-problem-2}
E^Q \big ( f(X_t)  \mid \F_t^{Z^N} \vee \F_t^{D^N} \big )
\end{equation}
where, with a slight abuse of notation,  ${Z}_t =\int_0^t{a}(X_s)ds+ {W}_t$ and
$ D_t =\sum_{ t_n\leq t} \delta_n (X_{t_n -} -K)^+.$
Note that    $\tau \wedge \sigma_N$ is the first exit time of $V$ from the domain $(K,N)$. Moreover, it holds  by definition  that $X_t =  V_{ t \wedge \tau \wedge \sigma_N}$, i.e.~$X$  is equal to   the asset value process $V$ stopped at the boundary of the bounded domain $(K,N)$.  Hence the state space of $X$ is given by $S^{X}: = [K,N]$ and the analysis of \citeasnoun{bib:pardoux-78} applies to the problem \eqref{eq:new-filtering-problem-2}.

In the next proposition we  show that the reduction to a bounded domain $(K,N)$, that is the use of the stopped process $V^N$ as underlying asset value process instead of the original process $V$,  does not affect the financial implications of the analysis, provided that $N$ is sufficiently large. In order to state the result we need to make the dependence of the model quantities on $N$ explicit: Let ${Z}_t^N =\int_0^t{a}(V_s^N)ds+ {W}_t$,  $D_t =\sum_{\{n\colon t_n\leq t\}} \delta_n (V^N_{t_n -} - K)^+$ and $\tau^N = \inf \{t \ge 0 \colon V_t^N \le K\}$, and denote by $\bbF^{\bbM,N}$ the modelling filtration in the model with asset value $V^N$, that is the filtration generated by $Z^N$, $D^N$ and by the default indicator $\ind{\tau^N \le t}$.

\begin{proposition} \label{prop:bounded-domain} 1. Fix some horizon date $T>0$ and let $\bbF$ be an arbitrary subfiltration of $\bbG$. Then  for $\epsilon >0$, it holds that
$$ Q \Big ( \sup_{0 \le t \le T} Q(\sigma_N \le t \mid \F_t ) > \epsilon \Big )  \le \frac{1}{\epsilon} Q( \sigma_N \le T) \to 0 \text{ as $N \to \infty$}. $$
2. The price process of the traded securities in the model with asset value process $V^N$ converges in ucp (uniformly on compacts in probability) to the price process in the model with asset value $V$. More precisely, consider  a function $h\colon [0,T] \times [K, \infty) \to \R $ such that $|h(t,v)| \le c_0 + c_1 v $. Then it holds that for $N \to \infty$,
\begin{equation} \label{eq:ucp-convergence}
 \sup_{0 \le t \le T}\Big |  \ind{\tau^N > t} E^Q \big ( h(t, V_t^N) \mid \F_t^{\bbM,N} \big)
 - \ind{\tau > t} E^Q \big ( h(t, V_t) \mid \F_t^{\bbM} \big)  \Big | \overset{Q}{\longrightarrow} 0\,.
 \end{equation}
\end{proposition}

\noindent The proof of the proposition is given in Appendix~B.

We continue with a few comments. Denote by $\bar \sigma_N = \inf\{t \ge 0: \bar{V}_t>N\}$ the first exit time of the cum-dividend asset value process from $(0, N)$. Clearly, $\bar \sigma_N \le \sigma_N$ as $V_t \le \bar{V_t}$ and thus $Q( \sigma_N \le T) \le Q(\bar \sigma_N \le T)$. Hence the conditional probability that $V$ reaches the upper boundary $N$ is  controlled uniformly for all subfiltrations $\bbF$ of $\bbG$ by the first exit time of a geometric Brownian motion from $(0,N)$; this can be used to choose $N$ when implementing of the model.
The ucp convergence  in the second statement ensures that the difference between the prices of traded securities in the  model based on $V^N$ and in the original model can be controlled uniformly in $t\in[0,T]$ which is stronger  than convergence in probability for fixed $t$.

\subsection{The case without dividends} \label{subsec:no-dividends}

In this section we consider the filtering problem \eqref{eq:new-filtering-problem-2} without  dividend
information; dividends will be included in Section~\ref{subsec:filtering-wrt-FI}.

\paragraph{Reference probability approach and Zakai equation.} As in \citeasnoun{bib:pardoux-78} we adopt the reference probability
approach to solve the problem \eqref{eq:new-filtering-problem-2}. Under this
approach one considers the model under a  so-called reference probability measure  $Q^*$ with $Q <\!\!< Q^*$  such that $Z$ and $X$ are independent under $Q^*$ and one reverts to the original dynamics via a change of measure. It will be convenient to model the pair $(X,Z)$ on a product space  $(\Omega,
\mathcal{G}, \bbG, Q^*)$. Denote by $(\Omega_2, \mathcal{G}^2,
\bbG^2,  Q_2)$ some filtered probability space that supports an $l$-dimensional Wiener process $Z = (Z_t(\omega_2))_{t \ge 0}$. Given some probability space  $(\Omega_1, \mathcal{G}^1,
\bbG^1 , Q_1)$ supporting the  process $X$  we let
$\Omega=\Omega_1\times \Omega_2$, $\mathcal{G} =\mathcal{G}^1 \otimes \mathcal{G}^2$,
$\bbG = \bbG^1 \otimes \bbG^2$ and  $Q^*= Q_1\otimes Q_2$, and we extend all processes to
the product space in the obvious way. Note that this construction implies that under
$Q^*$, $Z$ is an $l$-dimensional Brownian motion independent of $X$. Consider a
Girsanov-type measure transform of the form $L_t = ({d Q}/{dQ^*}){\vert_{\F_t}}$  with
\begin{equation}\label{eq:def-Lt1}
L_t =  L_t(\omega_1, \omega_2)=\exp \Big (\int_0^t a\big (X_s(\omega_1)\big)^\top d
Z_s(\omega_2) -\frac{1}{2}\int_0^t|a\big (X_s(\omega_1)\big)|^2 \,ds \Big )\,.
\end{equation}
Since $a$ is bounded $L$ is a true martingale by the Novikov criterion.
Girsanov's theorem for Brownian motion therefore  implies that under $ Q$ the pair $(X,Z)$
has the correct joint law. Using the abstract Bayes formula, one has for $f\in L^\infty
(S^X)$ that
\begin{equation} \label{eq:kallianpur-striebel}
E^{ Q} \big ( f(X_t) \mid \F_t^Z \big) = \frac{E^{Q^*} \big ( f(X_t)
L_t \mid \F_t^Z \big)}{E^{Q^*}\big(L_t \mid  \F_t^Z )}\,.
\end{equation}
We concentrate on the numerator. Using the product structure of the underlying
probability space we get that
\begin{equation} \label{eq:def-Sigma_t-1}
E^{Q^*} \big ( f(X_t) L_t \mid \F_t^Z \big)(\omega) = E^{Q_1} \big ( f(X_t)
L_t(\cdot, \omega_2)\big)=:\Sigma_t f  (\omega)\,.
\end{equation}
In Theorem~1.3 and 1.4 of \citeasnoun{bib:pardoux-78} the following characterization of
$\Sigma_t$ is derived.
\begin{proposition}\label{prop:Sigma-t}
Denote by $(T_t)_{t \ge 0}$ the transition semigroup of the Markov process $X$, that is
for $f \in L^\infty(S^X)$ and $x \in S^X$, $T_t f(x) = E_x^Q(f(X_t))$.  Then the following holds

1.)  $\Sigma_t f $ as defined in \eqref{eq:def-Sigma_t-1} satisfies the equation
\begin{equation} \label{eq:Zakai-mild-form}
\Sigma_t f = \Sigma_0 (T_t f) + \sum_{i=1}^l \int_0^t \Sigma_s (a^i \, T_{t-s} f)\, d Z_{s,i}
\end{equation}

2.)  Let $\tilde \Sigma$ be an $\bbF^Z$ adapted process taking values in the set of
    bounded and positive measures on $S^X$. Suppose that  for $f \in L^\infty(S^X)$
    $\tilde \Sigma_t f := \int_{S^X} f(x) \tilde \Sigma_t (dx)$ satisfies
    equation~\eqref{eq:Zakai-mild-form} and that  moreover  $\Sigma_0  = \tilde
    \Sigma_0 $. Then for all $0 \le t \le T$,  $\Sigma_t  =  \tilde \Sigma_t$ a.s.
\end{proposition}

%

\paragraph{An SPDE for the Density of $\Sigma_t$.}
Next we derive an SPDE for the density $u = u(t,\cdot)$ of the solution
$\Sigma_t$ of the Zakai equation~\eqref{eq:Zakai-mild-form}. We begin with the necessary
notation. First, we introduce the Sobolev spaces
$$
H^k(S^X) =\big \{ u \in L^2(S^X) \colon
\frac{d^\alpha u}{dx^\alpha}  \in L^2(S^X) \mbox{ for } \alpha \leq k \big \},
$$
where the derivatives are assumed to exist in the weak sense. Moreover, we let  $
H^1_0(S^X)=\{ u \in H^1(S^X) \colon u=0 \text{ on the boundary }
\partial S^X \}$. For precise definitions and further details on Sobolev spaces we refer to \citeasnoun{bib:adams-fournier-03}.   The scalar product
in $L^2(S^X)$ is denoted by $(\cdot\,,\cdot)_{S^X}$.
Consider for $f \in H^2 (S^X) $ the differential operator $\mathcal{L}^*$
with
\begin{equation} \label{eq:def-Lstar}
\mathcal{L}^* f (x) = \frac{1}{2} \frac{d^2}{dx^2} \big( \sigma^2 x^2 f\big) (x) - \frac{d}{dx}\big( r x f \big)(x).
\end{equation}
$\mathcal{L}^*$ is adjoint to $\mathcal{L}$ in the following sense: one has $\big ( f ,
\mathcal{L} g \big)_{S^X} = \big ( \mathcal{L}^* f ,g \big)_{S^X}${ whenever } $f,g \in
H^2 (S^X) \cap H^1_0(S^X)$.
Next we define an extension of $-\mathcal{L}^*$ to the entire space $H^1_0(S^X)$. For
this we denote by $H^1_0(S^X)'$ the dual space of $H^1_0(S^X)$ and by $\langle
\cdot,\cdot \rangle$ the duality pairing between $H^1_0(S^X)^\prime$ and $H^1_0(S^X)$.
Then we may define a bounded linear  operator $\mathcal{A}^*$ from $H^1_0(S^X)$ to
$H^1_0(S^X)^\prime $ by
\begin{equation} \label{eq:defA*}
\langle \mathcal{A}^* f, g \rangle = \frac{1}{2} \big (\sigma^2 x^2 \frac{df}{dx}\,,
 \frac{dg}{dx} \big )_{S^X} + \big( (\sigma^2 -r)x f\,,\,\frac{dg}{dx} \,\big
 )_{S^X}\,.
\end{equation}
Partial integration shows  that for $f \in H^2(S^X) \cap H^1_0(S^X)$ and $g\in
H^1_0(S^X)$ one has $\langle \mathcal{A}^* f, g \rangle = - \big ( \mathcal{L}^* f ,g
\big)_{S^X}$, so that $\mathcal{A}^*$ is in fact an extension of $-\mathcal{L}^*$.

We will show that the density of $\Sigma_t$ can be described in terms of the SPDE
\begin{equation} \label{eq:du-t-1}
du (t) = - \mathcal{A}^*u(t) dt + a^\top u(t) dZ_t, \quad u(0) = \pi_0,
\end{equation}
 This equation is to be understood as an equation in the dual space
$H^1_0(S^X)^\prime$, that is for every $v \in H^1_0(S^X)$ one has the relation
\begin{equation} \label{eq:d(u,v)-t-1}
\big ( u(t),v \big )_{S^X} = (u(0),v\big)_{S^X} - \int_0^t \langle \mathcal{A}^* u(s), v \rangle ds + \sum_{i=1}^l \int_0^t \big ( a^i u(s),v \big )_{S^X}
dZ_{s,i} \,.
\end{equation}
In the sequel we will mostly denote the stochastic integral with respect to the vector process $Z$ by   $ \int_0^t \big ( a^\top u(s),v \big )_{S^X} dZ_s $.

\begin{theorem} \label{thm:SPDE-for-u}
Suppose that Assumptions~\ref{ass:assets-and-dividends} and~\ref{ass:information} hold
and that the initial density $\pi_0$ belongs to $H^1_0(S^X)$. Then the following holds.
\vspace{0.05cm}

\noindent 1.~There is a unique $\bbF^Z$-adapted  solution $u \in L^2\big(\Omega \times
    [0,T], Q^*\otimes dt;  H^1_0(S^X)\big) $ of~\eqref{eq:du-t-1}.
\vspace{0.05cm}

\noindent 2.~The solution  $u$ has additional regularity: it holds that $u(t) \in H^2 (S^X)$ a.s.~and that the trajectories of $u$ belong to $C\big([0,T],H^1_0(S^X)\big)$, the space of $H^1_0(S^X)$-valued continuous functions with the supremum norm.
Moreover, $u(t,\cdot)\ge 0$ $Q^*$ a.s.
\vspace{0.05cm}

\noindent 3. The process $u(t)$  describes the  solution of the measure-valued
    Zakai equation \eqref{eq:Zakai-mild-form} in the following sense:  for $f \in L^\infty(S^X)$ one has
\begin{align}\label{eq:sigmat-in-terms-of-u}
\Sigma_t f &= \big ( u(t), f\big )_{S^X} + \nu_K (t) f(K) + \nu_N(t) f(N),\;  \text{ where }\\
0 &\le \nu_K(t) =  \int_0^t \frac{1}{2}\sigma^2 K^2 \frac{du}{dx}(s,K) ds , \\
0 &\le \nu_N(t)= -\int_0^t \frac{1}{2}\sigma^2 N^2 \frac{du}{dx}(s,N) ds + \int_0^t
a^\top(N) \nu_N(s) dZ_s \,.\label{eq:dynamics-nuN}
\end{align}
\end{theorem}
\paragraph{Comments.}  Since $u(t) $ belongs to  $H^2(S^X)\cap H^1_0(S^X)$,
\eqref{eq:d(u,v)-t-1} can be written as
\begin{equation} \label{eq:d(u,v)-t-2}
\big ( u(t),v \big )_{S^X} = (u(0),v\big)_{S^X} + \int_0^t \big ( \mathcal{L}^* u(s) , v \big)_{S^X}  ds
+ \int_0^t \big ( a^\top u(s),v \big )_{S^X}\, dZ_s;
\end{equation}
moreover, an approximation argument shows that \eqref{eq:d(u,v)-t-2} holds for $v \in L^2(S^X)$ (and not only for $v \in H_0^1(S^X)$).

 Statement 3 shows that the measure $\Sigma_t$ has a Lebesgue-density on the interior of
$S^X$ and a point mass on the boundary points $K$ and $N$. In view of
Proposition~\ref{prop:bounded-domain}, the point mass $\nu_N(t)$ is largely irrelevant;
the point mass $\nu_K(t)$ on the other hand will be important in the analysis of the
default intensity in Section~\ref{sec:dynamics}.

The assumption that $S^X$ is a bounded domain is needed in the proof of Statement 2; given the existence of a sufficiently regular nonnegative solution of equation \eqref{eq:du-t-1} the proof of Statement 3 is valid for an unbounded domain as well.

\begin{proof}
Statements~1 and 2 follow directly from Theorems~2.1, 2.3 and 2.6 of
\citeasnoun{bib:pardoux-78}. We give a sketch of the proof of the third claim,  as this  explains why
\eqref{eq:du-t-1}  is the appropriate SPDE to consider; moreover our arguments justify
the form of  $\nu_K$ and $\nu_N$.

The Sobolev  embedding  theorem (see for instance \citeasnoun{bib:adams-fournier-03}, Theorem~4.12, Part II and III) states that the space $H^m (S^X):= H^{m,2}(S^X)$ can be embedded into the H\"older space $\mathcal{C}^{k,\alpha}(S^X)$ for any $k\in \N, 0< \alpha<1$ such that $m-1/2 \ge k+\alpha$. It follows that $H^2(S^X)$ can be embedded into $\mathcal{C}^{1,\alpha}(S^X)$ for $0< \alpha<1/2$;
this ensures in particular that  the derivatives of $u$ at the boundary points of $S^X$ exist. Moreover, as $u(t,x) \ge 0$ on
$S^X$, we have $\frac{du}{dx}(t,K) \ge 0$ and thus $\nu_K(t) \ge 0$. Similarly, as $\frac{du}{dx}(t,N) \le 0$ we get from the standard comparison theorem for SDEs that
$\nu_N(t)$ is bigger than the solution $\tilde \nu$ of the SDE $\tilde \nu_t = \int_0^t a^\top(N) \tilde \nu_s dZ_s$. Now $\tilde \nu$ is clearly equal to zero so that  $\nu_N(t) \ge 0$ as well.

Denote by $\tilde \Sigma_t$ the measure-valued
process that is defined by the right side of \eqref{eq:sigmat-in-terms-of-u}. In order to
 show that $ \tilde \Sigma_t$ solves the mild-form Zakai equation
\eqref{eq:Zakai-mild-form}, fix some continuous function $f\colon S^X \to \R $ and some
$t \le T$, and denote by $\bar u(s,x)$ the solution of the terminal and boundary value problem
$$
\bar u_s + \mathcal{L} \bar u =0, \quad (s,x) \in (0,t) \times (K,N),
$$
with terminal  condition $\bar u(t,x)= f(x)$, $x \in S^X$, and boundary conditions $u(s,K) = f(K)$, $u(s,N) = f(N)$, $s \le t$. It is well-known that $\bar u$ describes the transition semigroup of
$X$, that is $\bar u(s,x) = T_{t-s} f(x)$, $0 \le s \le t$. As $\bar u(t) = f$ we obtain from the definition of $\tilde \Sigma_t $ and   the dynamics of $\nu_K(t)$ and  $\nu_N(t)$ that
\begin{align*}
\tilde \Sigma_t f =
 \big ( u(t) ,\bar u(t) \big )_{S^X} &+   \int_0^t \frac{1}{2}\sigma^2 K^2 \frac{du}{dx}(s,K)f(K) ds\\
 &- \int_0^t \frac{1}{2}\sigma^2 N^2 \frac{du}{dx}(s,N) f(N)ds + \int_0^t a^\top(N) \nu_N(s) f(N) \,dZ_s\,.
\end{align*}
Next we compute the differential of $ \big( u(t) ,\bar u(t) \big )_{S^X}$. We get, using the Ito product formula, \eqref{eq:d(u,v)-t-2} and the relation $d \bar u(s) = - \mathcal{L} \bar u(s) ds $, that
\begin{align*}
\big ( u(t),\bar u(t) \big )_{S^X} = (u(0), \bar u (0) \big)_{S^X} &+ \int_0^t \big ( \mathcal{L}^* u(s) , \bar u(s) \big)_{S^X}  ds
+ \int_0^t \big ( a^\top u(s),\bar u(s) \big )_{S^X} dZ_s\\
&+ \int_0^t \big ( u(s), - \mathcal{L} \bar u (s) \big )_{S^X} ds\,.
\end{align*}
Partial integration gives, using the boundary conditions satisfied by  $\bar u$,
$$ \int_0^t \big ( u(s), - \mathcal{L} \bar u (s) \big )_{S^X} ds = - \int_0^t \big ( \mathcal{L}^* u(s) , \bar u(s) \big)_{S^X} ds
+ \int_0^t \Big [\frac{1}{2}\sigma^2 x^2 \frac{du}{dx}(s,x)f(x)\Big ]_K^N \,ds\,.
$$
Hence we  get
$$
\tilde \Sigma_t f = \big ( u(0), \bar u(0) \big )_{S^X} +
\int_0^t \big (a^\top u(s),\bar u(s) \big )_{S^X} + a^\top(N) \nu_N(s) f(N) \, dZ_s \,.
$$
Now note that for $x \in [K,N]$, $\bar u (s) (x)= T_{t-s} f (x)$. Using that  $a(K) = 0$ by Assumption~\ref{ass:information}, we obtain that the stochastic integral with respect to $Z$ can be written as
$$
\int_0^t \Big \{ \big (a^\top u(s),T_{t-s} f  \big )_{S^X} +a^\top(K) (\nu_K(s) T_{t-s} f (K)) +  a^\top(N) (\nu_N(s) T_{t-s} f\,  (N)) \Big \}\, dZ_s\,.
$$
Hence it holds that $ \tilde \Sigma_t f = \tilde \Sigma_0 (T_t f) + \int_0^t \tilde \Sigma_s (a^\top T_{t-s} f)\, d Z_s\,.$
Moreover, $\Sigma_0 f = \big( \pi_0 , f\big )_{S^X} = \tilde \Sigma_0 f$. An approximation argument shows that these properties hold also for $f \in L^\infty (S^X)$, see for instance \citeasnoun{bib:pardoux-78}, so that $\Sigma_t = \tilde \Sigma _t $ by Proposition~\ref{prop:Sigma-t}.
\end{proof}

\begin{remark} \label{rem:krylov} It is interesting to  compare our results to the related paper \citeasnoun{bib:krylov-wang-11}.  Krylov and Wang  consider a  signal process $X$  that is a non-degenerate diffusion on $S^X$. Denoting  by $\tau_{S^X} $ the first exit time of $X$ from $S^X$ (in our notation $\tau_{S^X} =  \tau \wedge \sigma_N $), the observation filtration is given by $\bbF^Z$ and by the filtration generated by the  indicator $\ind{\tau_{S^X} \le t}$. Krylov and Wang then derive an SPDE for  the conditional density of $X_t$ given $\F_t^Z$ and the information $\{\tau_{S^X}>t\}$ and they show that
$$
Q(X_{\tau_{S^X}} =K \mid \tau=t) = \frac{\nu_K(t)}{\nu_N(t)+  \nu_K(t)}, \quad Q(X_{\tau_{S^X}} =N \mid \tau=t) = \frac{\nu_N(t)}{\nu_N(t)+  \nu_K(t)},
$$
where $\nu_K$ and $\nu_N$ are given by similar expressions as in Theorem~\ref{thm:SPDE-for-u}. However, they do not compute the dynamics of the conditional probability $Q(\tau \le t \mid \F_t^Z)$, an  expression that is crucial for the computation of default intensities (see Theorem~\ref{thm:default-intensity}).
\end{remark}

\subsection{Conditional Distribution with respect to $\bbF^\bbM$}
\label{subsec:filtering-wrt-FI}

In this subsection we compute the conditional distribution of $X$ with respect to the
filtration $\bbF^\bbM =\bbF^Z \vee\bbF^D\vee\bbF^Y$. The key part is to include the dividend information $\bbF^D$ and the jumps of the asset value process at the dividend dates in the analysis.  We recall some notation: the dividend dates are denoted by $t_n$, $n \ge 1$; $d_n$ denotes the dividend paid at $t_n$ and the conditional density of $d_n$ given $X_{t_n-}=x$ is  $\varphi(y, x)\Ind{\{x>K\}}$. In the sequel we let $t_0 =0$ for notational convenience. Moreover, we let $\varphi(y, K) =  \varphi^*(y)$ for some smooth and strictly positive reference density on $\R^+$ that we use in the construction of  the model via a change of measure.  Note that the  choice of $\varphi(y, K)$ has no economic implications, as we are only interested in the distribution of the asset value prior to default.

We use  an extension of the reference probability argument from  Section~\ref{subsec:preliminaries}. Consider a product space $\Omega=\Omega_1\times \Omega_2$, $\mathcal{G} =\mathcal{G}^1 \otimes \mathcal{G}^2$, $\bbG = \bbG^1 \otimes \bbG^2$ and  $Q^*= Q_1\otimes Q_2$ so that $\Omega_1$ supports a $Q_1$-Brownian motion $B$.  Suppose that $(\Omega_2, \mathcal{G}^2, \bbG^2 , Q_2)$ supports a a Brownian motion $Z$ and an independent random measure $\mu^D (dy,dt)$ with  compensating measure equal to
$$
\gamma^{D,*} (dy,dt) =  \sum_{n=1}^\infty \varphi^*(y) dy \delta_{\{t_n\}}(dt).
$$
Let $D_t := \int_0^t \int_{\R^+} y \, \mu^D(dy, dt) $, $t \ge 0$.  Denote by $V = V_t(\omega_1, \omega_2)$ the solution of the SDE $dV_t = \ind{V_t >0} r V_t  dt + \ind{V_t >0}\sigma V_t d B_t - \kappa d D_t$ and define the state process $X$ by $X_t = V_{t \wedge \tau \wedge \sigma_N}$.
The indicator function in the dynamics of $V_t$ is included as  under $Q^*$ the asset value  $V$ may become negative due to a downward jump at a  dividend date. Note that under  the measure $Q$ that we construct next such jumps have probability zero.

In order to revert to the original model dynamics we introduce the density martingale $L =   (L_{t}^1 L_{t}^2)_{0 \le t \le T}$ where $L^1_t$ is as in \eqref{eq:def-Lt1} and where $L_t^2= L_t^2 (\omega_1, \omega_2)$ satisfies
\begin{equation}\label{def:Lt2}
L_t^2  = 1+ \int_0^t \int_{\R^+} L_{s-}^2 \Big(  \frac{\varphi(y,X_{s-})}{\varphi^*(y)} - 1 \Big ) \,(\mu^D - \gamma^{D,*})(dy,ds) \,,\;
0 \le  t \le  T.
\end{equation}
Since $\varphi(\cdot,x)$ and $\varphi$ are probability densities we get
\begin{equation} \label{eq:int-wrt-gammaD}
\int_{\R^+} \Big(  \frac{\varphi(y,x)}{\varphi^*(y)} - 1 \Big )\varphi^*(y) dy = \int_{\R^+} (\varphi(y,x) - \varphi^*(y) ) dy = 1-1= 0\,.
\end{equation}

Hence $\int_0^t \int_{\R^+} \big(  \frac{\varphi(y,X_{s-})}{\varphi^*(y)} - 1 \big ) \, \gamma^{D,*}(dy,ds) \equiv 0$ and we obtain that
\begin{equation}\label{eq:dynamics-of-L2}
L_t^2  = 1+ \int_0^t \int_{\R^+} L_{s-}^2 \Big(  \frac{\varphi(y,X_{s-}
)}{\varphi^*(y)} - 1 \Big) \mu^D (dy,ds) = \prod_{t_n \le T} \frac{\varphi(d_n,X_{t_n -})}{\varphi^*(d_n)}\,.
\end{equation}
Since $L^1$ and $L^2$ are orthogonal we get that
$$dL_t = L_{t-}a(X_t)^\top d Z_t + \int_{\R^+} L_{t-} \big(  \frac{\varphi(y,X_{t-})}{\varphi^*(y)} - 1 \big ) \,(\mu^D - \gamma^{D,*})(dy,dt) \,.$$ The next lemma, whose proof  is given in   Appendix~B,  shows that $L$ is in fact the appropriate density martingale to consider ($T$ is the horizon date fixed at the beginning of Section~\ref{sec:filtering}).
\begin{lemma}\label{lemma:density-for-dividends}
It holds that $E^{Q^*}(L_T) =1.$
Define the measure $Q$ by $({dQ}/{dQ^*})\vert_{\G_T} = L_T$. Then under $Q$ the random measure $\mu^D$ has $\bbG$-compensator $\gamma^D(dy,dt) =  \sum_{n=1}^\infty \varphi(y,X_{t_n -}) dy \delta_{\{t_n\}}(dt)$. Moreover,  the triple $(X,Z,D)$  has the  joint law postulated in Assumption~\ref{ass:assets-and-dividends}.
\end{lemma}

Similarly as in \eqref{eq:kallianpur-striebel} we get from the generalized Bayes rule (\citeasnoun{bib:jacod-shiryaev-03}, Proposition~III.3.8) that
\begin{equation}\label{eq:kallianpur-striebel-2}
E^Q\left ( f(X_t) \mid \F_t^Z \vee \F_t^D \right )(\omega)  = \frac{\Sigma_t f (\omega_2) }{\Sigma_t 1 (\omega_2)}\,,
\end{equation}
where $\Sigma_t f\,
(\omega_2) = E^{Q^1} \big (f(X_t(\cdot, \omega_2 ) L_t(\cdot,\omega_2) \big) $.

\paragraph{Dynamics of the unnormalized density.} The form of $L_t$ in \eqref{eq:dynamics-of-L2} suggests the following dynamics of the unnormalized density $u(t,\cdot)$:  between dividend dates, that is on $(t_{n-1}, t_{n})$, $n \ge 1$,    $u(t)$ solves the SPDE \eqref{eq:du-t-1} with initial value $u(t_{n-1})$; at $t_{n}$ the density $u(t_n-)$ is first updated to
\begin{equation}\label{eq:jump-of-u}
\tilde u(t_n,x) = u(t_{n}-,x) \frac{\varphi(d_n,x)}{\varphi^*(d_n)}\,,
\end{equation}
second, for $\kappa =1$  there is a shift to account for the downward jump in the asset value, that is
\begin{equation}\label{eq:update-shift}
u(t_n,x) = \tilde u( t_n, x + \kappa d_n)\,,
\end{equation}
where we let $\tilde u (t_n,z) = 0$ for $z > N $. In Theorem~\ref{prop:dividend-information} below  we show that this is in fact correct. As a first step we describe the dynamics of $u$ by means of an SPDE. Denote for $y >0$  and  $v \in H^1_0 (S^X)$ by $S_y v$ the function $S_y v (x) = v(x+y)$, where we let $v(z) = 0$ for $z >N$. Consider the SPDE
\begin{equation} \label{eq:du-t-with-jumps}
du (t) = - \mathcal{A}^*u(t) dt + a^\top u(t) dZ_t + \int_{\R^+} \Big \{ S_{\kappa y} \Big  ( u(t-) \frac{\varphi(y,\cdot)}{\varphi^*(y)} \Big ) - u(t-) \Big \}\, \mu^D (dy,dt)\,,
\end{equation}
with initial condition $u(0) = \pi_0$. The interpretation of \eqref{eq:du-t-with-jumps} is analogous to the previous section: for  $v \in H^1_0(S^X)$ it holds that
\begin{align}\nonumber
\big ( u(t),v \big )_{S^X} &= (u(0),v\big)_{S^X} - \int_0^t \langle \mathcal{A}^* u(s), v \rangle ds
+ \int_0^t \big ( a^\top u(s),v \big )_{S^X}
dZ_s \\&+ \int_0^t \int_{\R^+} \big ( S_{\kappa y } \Big ( u(s-)  \frac{\varphi(y,\cdot)}{\varphi^*(y)}\Big )  - u(s-) \,, v \big)_{S^X}\, \mu^D (dy,ds)\,.\label{eq:int-wrt-gammaD^2}
\end{align}

The next result extends Theorem~\ref{thm:SPDE-for-u} to the case with dividends.

\begin{theorem} \label{prop:dividend-information}
1. There is a unique positive solution $u\in H_0^1(S^X)\cap H^2 (S^X) $ of the SPDE~\eqref{eq:du-t-with-jumps}.

\noindent 2.  Define $\nu_K(t) =  \int_0^t \frac{1}{2}\sigma^2 K^2 \frac{du}{dx}(s,K) ds$ and \begin{align*}
\nu_N(t) = &-\int_0^t \frac{1}{2}\sigma^2 N^2 \frac{du}{dx}(s,N) ds + \int_0^t
a^\top(N) \nu_N(s) dZ_s  + \int_0^t \int_{\R^+}\nu_N (s-)\big(\frac{\varphi(y,N)}{\varphi^*(y)} - 1 \big )\, \mu^D (dy,ds)\,.\nonumber
\end{align*}
Then it holds that $\Sigma_t f = \big ( u(t), f\big )_{S^X} + \nu_K (t) f(K) + \nu_N(t) f(N)$.
\end{theorem}
The proof is given in  Appendix~B.

\paragraph{Filtering with respect to $\bbF^\bbM $. } Finally we return to the filtering problem with respect to the  filtration  $\bbF^\bbM$.

\begin{corollary} \label{corr:cond-density}
Define the  norming  constant $C(t)$  by
$C(t) = \big ( u (t), 1\big )_{S^X} + \nu_N(t)$ and let $\pi(t,x) =  {u(t,x)}/{C(t)} $
and $\pi_N(t) = {\nu_N (t)}/{C(t)} $. Then it holds  for $f \in L^\infty(S^X)$ that
\begin{equation}
\ind{\tau >t}  E^Q \big (f(X_t) \mid \F_t^\bbM \big ) = \ind{\tau >t} \big ( ( \pi(t, \cdot), f)_{S^X}  + \pi_N(t) f(N) \big )\,.
\end{equation}
\end{corollary}
\begin{proof}
 Combining Proposition~\ref{prop:reduction-to-background} and Theorem~\ref{prop:dividend-information} we get
 \begin{equation}
 \ind{\tau >t}  E^Q \big (f(X_t) \mid \F_t^\bbM \big ) = \ind{\tau >t}\frac{\Sigma_t (f 1_{(K,\infty)})}{\Sigma_t   1_{(K,\infty)}} = \ind{\tau >t}\frac{\big ( u(t), f\big )_{S^X} + \nu_N(t) f(N) }{C(t)}.
 \end{equation}
 \end{proof}


\subsection{Finite-dimensional approximation of the filter equation}
\label{subsec:finite-dim-approx}

The SPDE \eqref{eq:du-t-1} is a  stochastic  partial differential equation and thus an infinite-dimensional object. In order to solve the filtering problem numerically and to generate price trajectories of  basic corporate securities one needs to approximate \eqref{eq:du-t-1} by  a finite-dimensional equation.
A natural way to achieve this is the Galerkin approximation method. We first explain the method for the case without dividend payments. Consider $m$ linearly independent  basis functions $e_1, \dots, e_m \in H^1_0(S^X) \cap H^2(S^X) $ generating the subspace  $\mathcal{H}^{(m)} \subset H^1_0(S^X) $,  and denote by $\prm \colon H^1_0(S^X) \to \mathcal{H}^{(m)}$ the projection on this subspace with respect to $(\cdot,\cdot)_{S^X}$. In the Galerkin method the solution $u^{(m)} $ of the equation
\begin{equation} \label{eq:galerkin-1}
d u^{(m)} (t) = \prm \circ \mathcal{L}^* \circ \prm u^{(m)} (t) dt + \prm ( a^\top \prm u^{(m)} (t))\, dZ_t
\end{equation}
with initial condition $u^{(m)}(0) = \prm \pi_0$ is used as an approximation to the solution $u$ of \eqref{eq:du-t-1}. Since projections are self-adjoint, we get that for $v \in H^1_0(S^X)$
\begin{equation} \label{eq:galerkin-2}
d \big( u^{(m)} (t),v \big)_{S^X} = \big ( \mathcal{L}^* \circ \prm u^{(m)} (t), \prm v \big )_{S^X} dt + \big ( a^\top \prm u^{(m)} (t), \prm v \big )_{S^X} d Z_t.
\end{equation}
Hence $d ( u^{(m)}  (t),v )_{S^X} =0$  if $v$  belongs to $ (\mathcal{H}^{(m)})^\bot $ (the orthogonal complement of $\mathcal{H}^{(m)}$). Since moreover $u^{(m)} (0) = \prm \pi_0 \in \mathcal{H}^{(m)}$ we conclude that $u^{(m)}(t) \in \mathcal{H}^{(m)}$ for all $t$. Hence $u^{(m)}  $ is of  the form
$u^{(m)} (t) = \sum_{i=1}^m \psi_i(t) e_i$, and we now determine an SDE system for the $m$ dimensional process $\Psi^{(m)}(t) = (\psi_1(t),\dots,\psi_m(t))^\prime$. Using \eqref{eq:galerkin-2} we get for $j \in \{1,\dots,m\}$
\begin{align} \label{eq:d psi-t-1}
d \big( u^{(m)}  (t),e_j \big)_{S^X} &= \sum_{i=1}^m \psi_i(t) \big (\mathcal{L}^* e_i, e_j \big )_{S^X} dt + \sum_{k=1}^l \sum_{i=1}^m  \big ( a_k e_i, e_j \big )_{S^X} \psi_i(t) d Z_t^k.
\intertext{On the other hand,}
 \label{eq:d psi-t-2}
 d  ( u^{(m)} (t), e_j )_{S^X} &= \sum_{i=1}^m (e_i,e_j) d \psi_i(t).
\end{align}
Define now the $m \times m$ matrices $A$, $B$ and  $ C^1, \dots, C^l$ with $a_{ij} = (e_i,e_j)_{S^X}$,  $b_{ij} =  (\mathcal{L}^* e_i, e_j  )_{S^X}$ and  $ c^k_{ij} =  ( a_k e_i, e_j  )_{S^X}$. Equating \eqref{eq:d psi-t-1} and \eqref{eq:d psi-t-2}, we get the following system of SDEs for $\Psi^{(m)}$
\begin{equation} \label{eq:SDE-for-Psi}
d \Psi^{(m)}(t) = A^{-1} B^\top \Psi^{(m)} (t) dt + \sum_{k=1}^l A^{-1}C^k \Psi^{(m)} (t) d Z_t^k
\end{equation}
with initial condition $\Psi^{(m)}(0) =  A^{-1} \big (
( \pi_0, e_1)_{S^X},\dots , (\pi_0, e_m)_{S^X}\big )^\prime$.
Equation~\eqref{eq:SDE-for-Psi} can be solved with numerical methods for SDEs such as a simple Euler scheme or the more advanced splitting up method proposed by \citeasnoun{bib:leGland-92}. Further  details regarding the numerical implementation of the Galerkin method are given among others in \citeasnoun{bib:frey-schmidt-xu-13} or in Chapter~4 of \citeasnoun{bib:roesler-16}.  Conditions for the convergence $u^{(m)}  \to u$ are well-understood, see for instance \citeasnoun{bib:germani-piccioni-87}: the Galerkin approximation for the filter density converges for $m \to \infty$ if and only if the  Galerkin approximation for the deterministic forward PDE $\frac{du}{dt} (t) = \mathcal{L}^* u (t)$ converges.

In the case with dividend information the  Galerkin method is applied successively on each interval $(t_{n-1}, t_{n})$, $n =1,2,\dots$. Denote by $ u_{n}^{(m)}$  the approximating density  over the interval  $(t_{n-1}, t_{n})$. Following \eqref{eq:du-t-with-jumps} the initial condition for the interval $(t_n, t_{n+1})$   is then given by
$$u^{(m)}(t_n) = \prm \left ( S_{\kappa y } \Big ( u^{(m)}_n( t_n ,\cdot) \frac{\varphi(d_n, \cdot)}{\varphi^*(d_n)}\Big ) \right )\,,$$
that is  by projecting the updated  and shifted density $ u^{(m)}_{n} (t_n , x + \kappa d_n) \big(\varphi(d_n, x + \kappa d_n )/ \varphi^*(d_n) \big )$ onto  $\mathcal{H}^{(m)} $.

\section{Dynamics of Corporate Security Prices}
\label{sec:dynamics}

In this section we identify  the price process  of traded corporate securities. It turns out
that these   price processes  are of
jump-diffusion type, driven by  a  Brownian motion $M^Z$ (the martingale part in the $\bbF^\bbM$ semimartingale decomposition of $Z$),   by the compensated random measure corresponding to the dividend payments and by the compensated default
indicator process.

\subsection{Default intensity}\label{subsec:default-intensity}

As a first step we derive the $\bbF^\bbM$-semimartingale decomposition of the default
indicator process  $Y$ and  we show that $Y$ admits an $\bbF^\bbM$-intensity.

\begin{theorem} \label{thm:default-intensity}
The  $\bbF^\bbM$-compensator of $Y$ is given by the process   $(\Lambda_{t\wedge \tau})_{t
\ge 0}$ where $\Lambda_t = \int_0^t  \lambda_{s-} ds $ and where the default intensity $\lambda_t$ is given by
\begin{equation} \label{eq:default-intensity}
 \lambda_t=
  \frac{1}{2}\sigma^2K^2 \frac{d \pi}{dx} (t,K)\, .
\end{equation}
Here  $\pi(t,x)$ is conditional  density of  $X_t$ given $\F_t^\bbM$ introduced in
Corollary~\ref{corr:cond-density}.
\end{theorem}
\noindent We mention that a similar result was obtained in \citeasnoun{bib:duffie-lando-01} for  the case where the noisy observation of the  asset value process  arrives only at deterministic time points.
\begin{proof}
We use the following well-known  result to determine the compensator of $Y$ (see for
instance Section~2.3 of \citeasnoun{bib:blanchet-scalliet-jeanblanc-04}).

\begin{proposition} \label{prop:compensator-via-hazard-approach}
Let $F_t = Q ( \tau \le t \mid \F_t^Z \vee \F_t^D)$ and suppose that $F_t < 1$  for all
$t$. Denote the Doob-Meyer decomposition of the bounded $\bbF^Z \vee \bbF^D$-submartingale $F$ by $F_t = M_t^F+ A_t^F$. Define the process $\Lambda$ via
$$ \Lambda_t = \int_0^t (1- F_{s-})^{-1} dA_s^F \,,\quad  t \ge 0.$$
Then
$Y_t - \Lambda_{t \wedge \tau}$ is an $\bbF^\bbM$-martingale. In particular, if  $A^F$ is
absolutely continuous, that is if $dA_t^F  = \gamma_t^A dt$, $\tau$ has the default
intensity $\lambda_t = \gamma_t^A/(1-F_{t-})$.
\end{proposition}
In order to apply the proposition we need to compute the Doob-Meyer decomposition of the
submartingale $F$.  Here we get
$$
F_t= Q ( \tau \le t \mid \F_t^Z \vee \F_t^D ) = Q (X_t= K \mid \F_t^Z \vee \F_t^D ) = \frac{\Sigma_t \ind{K} }{ \Sigma_t 1} .
$$
Theorem~\ref{prop:dividend-information} gives $\Sigma_t  \ind{K}  = \nu_K(t)$ and   $d \nu_K (t) =\frac{1}{2}\sigma^2 K^2 \frac{du}{dx}(t-, K) dt$.

Next we consider the term $ (\Sigma_t 1)^{-1} $.  By definition it holds that
$ \Sigma_t 1 = E^{Q^*}(L_t\mid \F_t^Z \vee \F_t^D) = ({dQ}/{d Q^* })\vert_{\F_t^Z \vee \F_t^D}.$  Hence we get that $(\Sigma_t 1 )^{-1} $ is a $Q$-local martingale; see for instance \citeasnoun{bib:jacod-shiryaev-03}, Corollary~III.3.10.
It\^o's product rule therefore gives that
$$
A_t^F = \int_0^t \frac{1}{\Sigma_{s-} 1 } \frac{1}{2}\sigma^2 K^2 \frac{du}{dx}(s-, K) \, ds.
$$
Furthermore we have
\begin{equation}\label{eq:1-F}
1-F_t = Q( X_t > K \mid \F_t^Z \vee \F_t^D) = \frac{1}{\Sigma_t 1} \big (\big(u(t), 1 \big)_{S^X}  + \nu_N(t) \big).
\end{equation}
The claim thus follows from Proposition~\ref{prop:compensator-via-hazard-approach} and
from the definition of $\pi(t,x) $ in Corollary~\ref{corr:cond-density}.
\end{proof}

\subsection{Asset Price Dynamics}
\label{subsec:filtering-dynamics}

In this section we derive the dynamics of the traded security prices.
In line with standard notation we denote for $f \colon ([0,T] \times S^X) \to \R$  with $E^Q(|f(t,X_t)|) < \infty$ for all $t \le T$ the optional projection of the process $(f(t,X_t))_{0 \le t \le T}$ on the modelling filtration by
$\widehat f_t  = E^Q(f(t,X_t) \mid \F_t^\bbM ) $. For smooth functions $f$ on $S^X$ we  define  the
operator $\mathcal{L}_X f(x)= \Ind{(K,N) }(x) \mathcal{L} f(x)$
($\mathcal{L}_X $ is the the generator of  $X$ between dividend dates).

Using Corollary~\ref{corr:cond-density} and the fact that $X_t = K$ on $\{\tau \le t\}$ one obviously has
\begin{equation} \label{eq:f-hat-vs-pitf}
\widehat{f}_t = \ind{\tau \le t} f(t,K) + \ind{\tau >t} \big (\pi(t),
f(t,\cdot)\big)_{S^X} + \pi_N(t) f(t,N) .
\end{equation}
Hence a crucial step in the derivation of asset price dynamics is to compute the dynamics of
$\pi_t f := \big (\pi(t),f(t,\cdot)\big)_{S^X} + \pi_N(t) f(t,N) $. This is done in the following proposition.

\begin{proposition} \label{lem:d-pit-f} With $\lambda_t = \frac{1}{2} \sigma^2 K^2 \frac{d
\pi(t,K)}{dx}$ it holds that
\begin{align} \label{eq:d-pit-f}
d \pi_t f &= \Big ( \pi_t \big (\frac{df}{dt} + \mathcal{L}_X f\big) - \lambda_t (f(t,K) - \pi_t f )\Big ) \, dt + \big (\pi_t( a^{\top} f) - \pi_t a^{\top}
\pi_t f \big) \, d (Z_t - \pi_t a \,dt)\\ \nonumber
 & +  \int_{\R^+} \Big(
\frac{\pi_{t-}({f(\cdot- \kappa y) \varphi(y,\cdot)})
}{ \pi_{t-} \varphi(y,\cdot)} - \pi_{t-} f \Big ) \; \mu^D(dy,ds)\,.
\end{align}
\end{proposition}
\noindent The proof is essentially a tedious application of the It{\^o} formula; it is given in  Appendix~B.

Now we are in a position to derive the price  dynamics of  the traded securities introduced in  Section~\ref{sec:pricing-basic-securities}. We begin with some notation. Let
\begin{equation}
\label{eq:def-^MZ}
M^Z_t  = M^{Z,\bbF^\bbM}_t = Z_t - \int_0^t \widehat{a}_s ds\, ,\quad t\ge 0.
\end{equation}
It is well known that  $M^Z$ is a $(Q, \bbF^\bbM)$ Brownian motion and hence the martingale
part in the $\bbF^\bbM$-semimartingale decomposition of $Z$. Next we define the $\bbF^\bbM$-martingale $M^Y$ by $M_t^Y = Y_t - \int_0^{t \wedge \tau} \lambda_s ds$. Finally, we
will  use the shorthand notation $(\widehat{\varphi(y)})_t$ for the optional projection of  $\varphi(y,X_t)$ on $\bbF^\bbM$   and we denote the $\bbF^\bbM$ compensator of $\mu^D$ by
\begin{equation}\label{eq:def-compensator-of-dividends-deterministic}
\gamma^{D, \bbF^{\bbM}}(dy, dt) = \sum_{n=1}^\infty (\widehat{\varphi(y)})_{t_n -}  dy \, \delta_{\{t_n\}} (dt) \,.
\end{equation}

\begin{theorem} \label{thm:asset-price-dynamics} Denote by $\Pi^{\text{surv}}$, by $\Pi^{\text{def}}$ and by  $S$ the ex-dividend price (the price value  of the future cash flow stream) of the survival claim, of the default claim and of the stock  of the firm. Then it holds that
\begin{align} \label{eq:dHsurv}
\Pi_t^{\text{surv}}  & = \Pi_0^{\text{surv}}  + \int_0^{t \wedge \tau}  r \Pi_s^{\text{surv}} ds  +  \int_0^{\tau \wedge t} (\widehat{ h^{\text{surv}} a^\top})_{s-} -  \Pi^{\text{surv}}_{s-} \,\widehat {a}_{s-}^{\top} \, d M_s^{Z,\bbF^\bbM} \\\nonumber & - \int_0^{t \wedge \tau}    \Pi^{\text{surv}}_{s-} \,d  M_s^Y +
 \int_0^{\tau \wedge t} \int_{\R^+} \frac{( \widehat{h^{\text{surv}} \varphi(y)})_{s-}}{ (\widehat{\varphi (y)})_{s-}} -    \Pi^{\text{surv}}_{s-} \;
  (\mu^D -\gamma^{D,\bbF^{\bbM}})(dy,ds )  \,.\\ \label{eq:dHdef}
 \Pi_t^{\text{def}}  & = \Pi_0^{\text{def}}  + \int_0^{t \wedge \tau}  r \Pi_s^{\text{def}} -  \lambda_s ds  +  \int_0^{\tau \wedge t} (\widehat{ h^{\text{def}} a^\top})_{s-} -  \Pi^{\text{def}}_{s-} \,\widehat {a}_{s-}^{\top} \, d M_s^{Z,\bbF^\bbM}
 \\\nonumber & -  \int_0^{t \wedge \tau}   \Pi^{\text{def}}_{s-} \,d  M_s^Y +
 \int_0^{\tau \wedge t} \int_{\R^+} \frac{( \widehat{h^{\text{def}} \varphi(y)})_{s-}}{ (\widehat{\varphi (y)})_{s-}} -    \Pi^{\text{def}}_{s-} \;
  (\mu^D -\gamma^{D,\bbF^{\bbM}})(dy,ds )   \\
\label{eq:dSt}
S_t  & =  S_0  + \int_0^{t \wedge \tau}  r{S}_s ds  - \int_0^{t \wedge \tau} \int_{R^+} y \gamma^{D,\bbF^{\bbM}}(dy,ds ) +  \int_0^{\tau \wedge t} (\widehat{ h^{\text{stock}} a^\top})_{s-} -  {S}_{s-} \,\widehat {a}_{s-}^{\top} \, d M_s^{Z,\bbF^\bbM} \\& - \int_0^{t \wedge \tau} {S}_{s-} \,d  M_s^Y +
 \int_0^{\tau \wedge t} \! \int_{\R^+} \frac{( \widehat{h^{\text{stock}} \varphi(y)})_{s-} }{ (\widehat{\varphi (y)})_{s-}}- S_{s-} \; (\mu^D -\gamma^{D,\bbF^{\bbM}})(dy,ds )  \,. \nonumber
\end{align}
\end{theorem}
\begin{proof}
We begin with the survival claim. It follows from relations~\eqref{eq:pricing-via-filtering} and~\eqref{eq:f-hat-vs-pitf} that
$$
\Pi_t^{\text{surv}} = 1_{\{\tau >t\}} (\widehat{h^{\text{surv}}})_t =  1_{\{\tau >t\}} \pi_t h^{\text{surv}}
$$
so that $ d \Pi_t^{\text{surv}}  = (1- Y_{t-}) d \pi_t h^{\text{surv}} - \Pi_{t-}^{\text{surv}}dY_t$. Now recall that
$\frac{d }{dt} h^{\text{surv}} + \mathcal{L}_X h^{\text{surv}} = r h^{\text{surv}}$ and that  $h^{\text{surv}}(t,K) \equiv 0$. Substituting these relation into the dynamics of $\pi_t h^{\text{surv}}$ gives
\begin{align}\nonumber
d \pi_t h^{\text{surv}} &= \Big ( r \pi_t h^{\text{surv}}  +  \lambda_t \pi_t h^{\text{surv}} \Big ) \, dt + \big (\pi_t( a^{\top} h^{\text{surv}}) - \pi_t a^{\top}
\pi_t h^{\text{surv}} \big) \, d (Z_t - \pi_t a \,dt)\\ \label{eq:hsurv-int-wrt-muD}
 & +  \int_{\R^+} \Big(
\frac{\pi_{t-}({h^{\text{surv}}(\cdot- \kappa y) \varphi(y,\cdot)})
}{ \pi_{t-} \varphi(y,\cdot)} - \pi_{t-} h^{\text{surv}} \Big ) \; \mu^D(dy,dt)\,.
\end{align}
 Now, using the definition of $\gamma^{D,\bbF^{\bbM}}$ and Fubini we get at a dividend date $t_n < \tau$  that
\begin{align} \label{eq:hsurv-int-wrt-muD-2}
 &\int_{\R^+} \frac{\pi_{t-}({h^{\text{surv}}(\cdot- \kappa y) \varphi(y,\cdot)})
   }{ \pi_{t-} \varphi(y,\cdot)}  \; \gamma^{D,\bbF^{\bbM}}(dy,\{t_n\}) =
 \pi_{t_n-} \Big( \int_{\R^+}  ({h^{\text{surv}}(\cdot- \kappa y) \varphi(y)})_{t_n-} dy\Big)\,.
 \end{align}
Relation~\eqref{eq:full-info-value-at-tn-surv} implies that the right hand side of \eqref{eq:hsurv-int-wrt-muD-2}  is equal to
$\pi_{t_n-} {h^{\text{surv}} (t_n, \cdot)}$. This shows that in \eqref{eq:hsurv-int-wrt-muD}
the integral with respect to $\mu^D(dy,ds)$ can be replaced with an integral with respect to  $(\mu^D -\gamma^{D,\bbF^{\bbM}})(dy,ds )$. Since for generic functions $f \colon [0,T] \times S^X \to \R $ it holds that $\widehat{f}_t = \pi_t f$ on  $\{t < \tau \} $  we finally  obtain the result for $\Pi^\text{surv}$.
Mutatis mutandis these arguments also apply to the default claim and to the stock price. The additional term $-\lambda_s ds$ in the drift of $\Pi^{\text{def}}$ stems from the fact that $h^{\text{def}}(t,K) =1$;  the additional integral with respect to $\gamma^{D,\bbF^{\bbM}}(dy, ds)$ in the dynamics of the stock price is due to the different behaviour of $h^{\text{stock}}$ at a dividend date, see~\eqref{eq:full-info-value-at-tn}. Of course this term  is quite intuitive: the expected downward jump in the stock price at a dividend date is just equal to the expected dividend payment.
\end{proof}

\paragraph{Comments and extensions.} Theorem~\ref{thm:asset-price-dynamics} formalizes the idea that the prices of traded corporate securities are driven by the arrival of new information on the value of the underlying firm,  since the processes that drive the  asset price dynamics are closely related to the generators of $\bbF^\bbM$.

In order to study dynamic  hedging strategies we need the dynamics and  the predictable quadratic variation of the  cum dividend price or \emph{gains process} of the traded assets. The survival claim has no intermediate cash flows and we have $dG^{\text{surv}}_t = d\Pi^{\text{surv}}_t$; for the default claim it holds that $dG_t^{\text{def}} = d\Pi^{\text{def}}_t + dY_t$; for the stock we have $dG_t^{\text{stock}}= dS_t + (1-Y_{t-})\,dD_{t}$. Note that Theorem~\ref{thm:asset-price-dynamics} implies that  the discounted gains processes of all three assets are martingales --- as they have to be given that we work directly under a  martingale measure  $Q$.
To compute the   quadratic variations note that from Theorem~\ref{thm:asset-price-dynamics}, the  discounted gains process of the $i$th traded asset has a martingale representation of the form
\begin{align*}
\widetilde{G}_t^i &= G_0^i + \int_0^{t\wedge \tau}\!(\xi_{s,i}^{M^Z})^\top d M_s^Z + \int_0^{t\wedge \tau}\!\xi_{s,i}^{Y}\,d M_s^Y +
\int_0^{t\wedge \tau} \!\!\! \int_{\R^+} \!\xi_{i}^{D} (s,y) (\mu^D -  \gamma^{D,\bbF^{\bbM}})(dy,ds),
\end{align*}
and the integrands are explicitly given in  the theorem. Define a measure $b$ on $[0, \infty)$ by letting $b([0,t]) = b(t):= t + \sum_{n =1}^\infty \delta_{\{t_n\}} ( [0,t]) $ ($b$ is the sum of the Lebesgue measure and the counting measure on the set  of dividend dates $\mathcal{T}^D$).   Then the predictable quadratic variation with respect to $\bbF^\bbM$ of the discounted gains processes of asset $i$ and asset $j$ is of the form
$\langle {\widetilde G}^i, \widetilde{G}^j \rangle_t  =  \int_0^{t\wedge \tau} v_s^{ij} d b(s)$ with \emph{instantaneous quadratic variation} $v^{ij}_s$ given by
\begin{align} \label{eq:quad-var}
 v^{ij}_s &= 1_{([0,\infty) \setminus \mathcal{T}^D )}(s) \Big ((\xi_{s,i}^{M^Z})^\top (\xi_{s,j}^{M^Z}) + \xi_{s,i}^{Y} \xi_{s,j}^Y \lambda_s \Big ) +
 1_{\mathcal{T}^D}(s) \int_{\R^+} \xi_{i}^{D} (s,y) \xi_{j}^{D} (s,y) (\widehat{\varphi(y)})_{s-} \, dy \,.
\end{align}


\section{Derivative Asset Analysis} \label{sec:applications}

 In this section we discuss the pricing and  the  hedging of   securities related to the firm  that are \emph{not} liquidly traded such as  bonds with non-standard maturities or options on  the traded assets.   We assume that the risk-neutral pricing formula \eqref{eq:def-price-of-H}
applies also to non-traded securities so that  the price  at time $t$ of a security with $\F_T^\bbM$-measurable integrable payoff $H$ is given by
\begin{equation} \label{eq:risk-neutral-pricing-2}
\Pi_t^H = E^Q \big ( e^{- r(T-t)} H \mid \F_t^\bbM \big ) .
\end{equation}
Note that while very natural in our framework, \eqref{eq:risk-neutral-pricing-2}  is in fact an assumption. In our model markets are typically  not complete so that the martingale measure is not unique and an ad hoc assumption on the choice of the pricing measure has to be made. This is an unpleasant but unavoidable feature of most models where  asset prices follow diffusion processes  with jumps. A second issue with \eqref{eq:risk-neutral-pricing-2} is the fact that prices are defined as conditional expectations with respect to the fictitious modeling filtration $\bbF^\bbM$, whereas prices should be computable in terms of quantities that are observable for the model user.
In Section~\ref{subsec:derivatives} we therefore show that for the derivatives  common in practice, $\Pi_t^H$ defined in \eqref{eq:risk-neutral-pricing-2} is given by a function $C^H(t, \pi(t))$ of time and the current filter density $\pi(t)$ and we discuss the evaluation of $C^H$. In  Section~\ref{subsec:calibration}  we moreover explain how to determine an estimate of $\pi(t)$ from  prices of traded securities observed at time $t$.
Section~\ref{subsec:hedging} is  concerned with risk-minimizing  hedging strategies.

\subsection{Derivative Pricing} \label{subsec:derivatives}

Most derivative securities related to the firm fall in one of the following two classes.

\paragraph{Basic debt securities.} Examples of non-traded basic debt securities are bonds or CDSs with non-standard maturities.
The pricing of these  securities is straightforward.   Let $h$ be the full information value of the security.    A similar argument as in Section~\ref{sec:pricing-basic-securities} shows that
$$ 1_{\{\tau > t\}} \Pi_t^H = 1_{\{\tau > t\}}  E^Q \big (h(t,V_t) \mid \F_t^\bbM) =  \int_K^\infty h(t,v)\pi(t,v) d v\, , $$
that is $\Pi_t^H $ can be computed by averaging the full-information value with respect to the current filter density $\pi(t)$ (which is determined by calibrating the model to the prices of traded  securities, see Section~\ref{subsec:calibration}).

\paragraph{Options on traded assets. } In its most general form the payoff of an option on a traded asset with maturity $T$ is of the form
$H = g(\Pi_T^1,\dots,  \Pi_T^\ell)$ where  $\Pi^1,\dots,  \Pi^\ell$ are  the ex-dividend price processes of
$\ell$ traded risky assets related to the  the firm.  Examples for such products  include equity  and bond options or certain convertible bonds. Note that $H$ is $\F_T^\bbM$-measurable since the rvs $\Pi_T^1, \dots, \Pi_T^\ell$ are $\F_T^\bbM$- measurable by \eqref{eq:def-price-of-H}.

Our goal is to show that the price of an option on traded assets can be written as a function of the current filter density $\pi(t)$.
We consider  an option on the stock with  payoff $ H =  g({S}_T)$;  other options can be handled with only notational changes.  We get for the price of the option  that
\begin{align*} \nonumber
\Pi_t^H & =  E^Q \big ( e^{-r(T-t)} \ind{\tau > T} g ( {S}_T) \mid \F_t^\bbM \big) + e^{-r(T-t)} g(0) Q(\tau \le T\mid \F_t^\bbM).
\end{align*}
The second term is the price of a basic debt security. In order to deal with the first term we now give a general result that shows that the computation of $ E^Q \big ( e^{-r(T-t)} \ind{\tau > T} g ( {S}_T) \mid \F_t^\bbM \big)$  can be reduced to the problem of computing a conditional expectation with respect to the reference measure $Q^*$ and the $\sigma$ field  $\F^Z_t \vee \F_t^D$ from the background filtration.
\begin{lemma}\label{lemma:survival-claim}
Consider some integrable, $\F^Z_T \vee \F_T^D$ measurable random variable $ H$ such as $H = g(S_T)$. Then it holds for $t \le T$ that
\begin{equation}\label{eq:pricing-survival-claim}
E^Q\big (  \ind{\tau >T}  H \mid \F_t^\bbM \big ) =  \ind{\tau >t} \frac{E^{Q^*} \big ( H \, \big ((u(T),1)_{S^X }+ \nu_N(T) \big)  \mid \F^Z_t \vee \F_t^D \big )}{(u(t),1)_{S^X } + \nu_N(t)}\,.
\end{equation}
\end{lemma}

\begin{proof}   As in the proof of Theorem~\ref{thm:default-intensity} we let $F_t= Q ( \tau \le t \mid \F_t^Z \vee \F_t^D )$. Then the Dellacherie formula gives
\begin{equation}\label{eq:applic-of-dellacherie}
E^Q\big (  \ind{\tau > T} H \mid \F_t^\bbM \big ) =  \ind{\tau >t} \frac{E^Q \big( (1-F_T)  H   \mid \F^Z_t \vee \F_t^D \big )}{1-F_t}.
\end{equation}
Since for generic $s \in [0,T]$ one has  $\Sigma_s 1 = \frac{dQ}{dQ^*}\vert_{\F^Z_s \vee \F_s^D}$,  the abstract Bayes formula yields
$$
 E^Q \big( (1-F_T)  H   \mid \F^Z_t \vee \F_t^D \big ) = \frac{1}{\Sigma_t 1} E^{Q^*} \Big( (\Sigma_T 1)  \, (1-F_T)  H   \mid \F^Z_t \vee \F_t^D \Big )\,.
$$
Moreover, using  \eqref{eq:1-F} we have for  $s \in [0,T]$   that
$
(\Sigma_s 1) (1-F_s)   =   \big (\big(u(s), 1 \big)_{S^X}  + \nu_N(s) \big)\,.
$
Substituting these  relations into \eqref{eq:applic-of-dellacherie} gives the result.
\end{proof}

Now we return to the stock option. For simplicity we ignore the point mass $\nu_N$ at the upper boundary  of $S^X$.  Recall that ${S}_T =(u(T),h^{\text{stock}})_{S^X } \big/ (u(T),1)_{S^X } $ Using Lemma~\ref{lemma:survival-claim} we get that
\begin{equation*}\label{eq:price-of-options1}
 E^Q \big ( e^{-r(T-t)} \ind{\tau > T} g ( {S}_T) \mid \F_t^\bbM \big) = \ind{\tau >t} \frac{  E^{Q^*} \Big (
g \Big(\frac{ (u(T),h^{\text{stock}})_{S^X }}{(u(T),1)_{S^X }}\Big )\, (u(T),1 )_{S^X } \mid \F^Z_t \vee \F_t^D \Big )
}{(u(t),1)_{S^X} }\,,
\end{equation*}
Standard results on the Markov property of solutions of SPDEs such as Theorem 9.30 of \citeasnoun{bib:peszat-zabcyk-07} imply that under $Q^*$ the solution $u(t) $ of the SPDE \eqref{eq:du-t-with-jumps} is a Markov process.
Hence
\begin{equation}\label{eq:price-of-options2}
\frac{1}{(u(t),1)_{S^X }} E^{Q^*} \Big (g \Big(\frac{(u(T),h^{\text{stock}})_{S^X }}{(u(T),1)_{S^X }}\Big )\, (u(T),1 )_{S^X } \mid \F^Z_t \vee \F_t^D \Big ) =  C^H(t,u(t))
\end{equation}
for some function   $C^H $  of time and of the current value of the unnormalized filter  density. Moreover, $C^H$ is homogeneous of degree zero in $u(t)$, as we now explain.   Since the the SPDE \eqref{eq:du-t-with-jumps} is linear,  the solution  of \eqref{eq:du-t-with-jumps} over the time interval $[t,T]$ with initial condition $\gamma u(t)$ ($\gamma >0$ a given constant) is given by $\gamma u(s)$, $s \in [t, T]$. If we substitute this into \eqref{eq:price-of-options2} we get that $C^H(t,\gamma u(t)) =  C^H(t,u(t))$ as $\gamma$ cancels out.
Hence we may without loss of generality replace $u(t)$ by the current filter density $\pi(t) = u(t)\big/(u(t),1)_{S^X}$, and we get
\begin{equation} \label{eq:price-of-options-final}
E^Q \big ( e^{-r(T-t)} \ind{\tau> T} g ( {S}_T) \mid \F_t^\bbM \big) =\ind{\tau >t} C^H(t, \pi(t))\,.
\end{equation}

The actual computation of $C^H$ is best done using Monte Carlo simulation, using a numerical method to solve the SPDE~\eqref{eq:du-t-with-jumps}. The Galerkin approximation described in Section~\ref{subsec:finite-dim-approx} is particularly well-suited for this purpose since most of the time-consuming computational steps can be done off-line. Note that \eqref{eq:price-of-options-final} is an expectation with respect to the reference measure $Q^*$. Hence  one needs to sample from the SDE \eqref{eq:du-t-with-jumps} under $Q^*$, that is the driving process $Z$ is a Brownian motion and the random measure $\mu^D$ has compensator $\gamma^{D,*}(dy,dt). $. Alternatively, one might evaluate directly the expected value $ E^Q \big ( e^{-r(T-t)} \ind{\tau > T} g ( {S}_T) \mid \F_t^\bbM \big)$, using the simulation approach sketched in Section~\ref{sec:simulations} below.


\subsection{Hedging}\label{subsec:hedging}

Hedging is a key aspect of derivative asset analysis.  In this section we therefore use our results on the price dynamics of traded securities to derive dynamic hedging strategies.  We expect the market to be incomplete, as the prices of the traded securities follow diffusion processes with jumps. In order to deal with
this problem we use the concept of {risk minimization} introduced by \citeasnoun{bib:foellmer-sondermann-86}.   A similar analysis was carried out in~\citeasnoun{bib:frey-schmidt-12} in the context of reduced-form credit risk models.

\paragraph{Risk minimization.} We first introduce the notion of a risk-minimizing hedging strategy.
We assume that there are $\ell$ traded securities related to the firm with  ex-dividend price process $\Pi = (\Pi_t^1,\dots,  \Pi_t^\ell)_{t \le T}^\top$ and gains processes ${G} = ({G}_t^1, \dots, {G}_t^\ell)^\top_{t \le T}$; moreover there is a continuously compounded money market  account with  value $e^{r t}$, $t \ge 0$. The discounted price and gains processes are denoted by $\widetilde{\Pi}$ and  $\widetilde{G}$. Recall that the predictable quadratic variation of the gains process of the traded assets is of the form  $\langle {\widetilde G}^i, \widetilde{G}^j \rangle_t  =  \int_0^{t\wedge \tau} v_s^{ij} d b(s)$ with $v^{ij}$ and $b$ given in Section~\ref{subsec:filtering-dynamics} (see equation~\eqref{eq:quad-var}) and let $\mathbf{v}_t = (v^{ij}_t)_{1 \le i,j \le \ell}$.    Denote by
$L^2(\widetilde G^1,\dots, \widetilde G^n , \bbF^\bbM)$ the space of all $\ell$-dimensional
$\bbF^\bbM$-predictable processes $\theta$ such that $E \big( \int_0^T \theta_s^\top
\mathbf{v}_s \theta_s ds \big ) < \infty.$

An \emph{admissible trading  strategy} is given by a pair $\phi=(\theta,\eta)$ where
$\theta\in L^2(\widetilde G^1,\dots, \widetilde G^n , \bbF^\bbM)$ and  $\eta$ is $\bbF^\bbM$-adapted; $\theta_t $ gives the position in the risky assets at time $t$ and $\eta_t$ the position in the  money market account. The value  of this strategy at time $t$ is $V_t^\phi = \theta_t^\top \Pi_t + \eta_t e^{rt}$ and the discounted value is $\widetilde{V}_t^\Phi = \theta_t^\top \widetilde{\Pi}_t + \eta_t$. In the sequel we  consider strategies whose value tracks a given stochastic process. In an incomplete market this is only feasible if we allow for   intermediate in-and outflows of  cash. The size of these in-and outflows is measured by the discounted \emph{cost process} $C^\phi$ with $C_t^\phi = \widetilde{V}_t^\phi - \int_0^t \theta^\top d \widetilde{G}_s$. We get that
$$
    C_T^\phi - C_t^\phi =
    \widetilde{V}_T^\phi - \int_0^T \theta_s^\top d \widetilde{G}_s - \Big ( \widetilde{V}_t^\phi - \int_0^t \theta_s^\top d \widetilde{G}_s \Big )
    = \widetilde{V}_T^\phi - \Big ( \widetilde{V}_t^\phi + \int_t^T \theta_s^\top d \widetilde{G}_s \Big ),
$$
that is $  C_T^\phi - C_t^\phi $ gives the cumulative capital injections or withdrawals over the period $(t, T]$. In particular, for a selffinancing strategy it holds that $ C_T^\phi - C_t^\phi =0$ for all $t$. Finally we define the \emph{remaining risk process} $R(\phi)$ of the strategy by
\begin{equation}
R_t(\phi) = E \big( (C_T - C_t)^2 | \F^\bbM_t\big),\quad 0 \le t \le T .
\end{equation}
Consider now  a claim  with square integrable $\F_T^\bbM$-measurable
payoff $H$  and an admissible strategy
$\phi$ with $V_T^\phi= H$ (note that  this condition can always be achieved by a proper choice of the cash position $\eta_T$). Then $R(\phi)$ is a measure for the precision of the hedge, in particular, $R(\phi) \equiv 0$ if $\phi$ is a selffinancing hedging strategy for $H$. An admissible  strategy $\phi^*$ is called \emph{risk-minimizing} if $V_T^{\phi^*}= H$ and if moreover for any $t \in [0,T]$ and any admissible strategy $\phi$ satisfying $V_T^{\phi} = H$ we have $R_t (\phi^*) \le
R_t(\phi)$.
Risk-minimization  is well-suited for our setup as the ensuing hedging strategies are
relatively easy to compute and as it suffices to know the risk-neutral dynamics of the
traded securities.\footnote{It might be more natural
to minimize the remaining risk under the historical probability measure. This would lead
to alternative quadratic-hedging approaches; see for instance~\citeasnoun{bib:schweizer-01b}.
However, the computation of the corresponding strategies becomes a very challenging
problem.}

Next we  give a general characterization of risk-minimizing hedging strategies.  Let $\Pi_t^H = E^Q(e^{-r(T-t)} H \mid \F_t^\bbM)$ so that the discounted price process $\widetilde{\Pi}^H$ is a square integrable $\bbF^\bbM$ martingale. It is well-known that the predictable covariation
$\langle \Pi^H, \widetilde{G}^i\rangle $ is absolutely continuous with respect to $\langle \widetilde{G}^i\rangle $ and hence with respect to the measure $b$ introduced before \eqref{eq:quad-var}, and we denote the density by $\nicefrac{d\langle \Pi^H, \widetilde{G}^i\rangle} { db } (t)$; finally $ \nicefrac{d\langle \widetilde{\Pi}^H, \widetilde{G}\rangle }{ db} \,(t)$ stands for the vector process $\nicefrac{\big(d\langle \widetilde{\Pi}^H, \widetilde{G}^1\rangle}{  db} \,(t), \dots, \nicefrac{d\langle \widetilde{\Pi}^H, \widetilde{G}^\ell\rangle}{db} \,(t)\big )^\top $.

\begin{proposition}\label{prop:hedging}
A risk-minimizing strategy $\phi^* =(\theta^*, \eta^*)$ for a claim $H \in L^2(\Omega, \F_T^\bbM, Q)$ exists. It can be characterized as follows:  $\theta_t^*$ is a solution of the equation $ \mathbf{v}_t \theta_t^* = \displaystyle{\nicefrac{d\langle\widetilde{\Pi}^H ,\widetilde G \rangle}{d b}}\,(t)$; the cash position is $\eta_t^* = \widetilde{\Pi}_t^H - (\theta^*_t)^\top \widetilde \Pi_t$ and it holds that $V_t^{\phi^*} = \Pi^H_t$.
\end{proposition}

\begin{proof}
First we recall the \emph{Kunita Watanabe decomposition} of the martingale $\widetilde{\Pi}^H$ with respect to the  gains processes of the  traded securities.  This decomposition is given  by
\begin{align} \label{eq:GKW-decomp}
\widetilde{\Pi}^H_t &= \widetilde{\Pi}^H_0 +  \sum_{i=1}^\ell \int_0^t \xi^H_{s,i} \,   d\widetilde{G}_{s}^i+ H^\bot_t, \quad 0 \le t \le T;
\end{align}
here $\xi_i^H \in L^2(\widetilde G^1,\dots, \widetilde G^n , \bbF^\bbM)$ and  the martingale $H^\bot$ is strongly orthogonal to the gains processes of the traded securities, that is $\langle H^\bot, \widetilde{G}^i\rangle \equiv 0$ for all $1 \le i \le \ell$.
As shown in \citeasnoun{bib:foellmer-sondermann-86},   risk-minimizing hedging strategies relate to the Kunita Watanabe decomposition~\eqref{eq:GKW-decomp} as follows: it holds that   $\theta^*=\xi^H$, that $\widetilde{V}^\phi =\widetilde{\Pi}^H $ and that $C=H^\bot$.
Next we identify $\theta^*$. As  $\langle H^\bot, \widetilde{G}^i\rangle \equiv 0$, the Kunita Watanabe decomposition gives
$ \langle \widetilde{\Pi}^H, \widetilde{G}^i \rangle_t = \sum_{j=1}^\ell \int_0^t \theta^*_{s,j} d \langle \widetilde{G}^j, \widetilde{G}^i \rangle_s
$
or equivalently
$$\int_0^{t \wedge \tau} \frac{d \langle \widetilde{\Pi}^H, \widetilde{G}^i \rangle}{db} (s)\, d b(s) =  \int_0^{t \wedge \tau}  \sum_{j=1}^\ell \theta^*_{s,j} v_s^{ji}\, d b(s),
$$
which shows that $ \mathbf{v}_t \theta_t^* = \frac{d\langle\widetilde{\Pi}^H, \widetilde G \rangle}{d b}(t)$. The remaining statements are  clear.
\end{proof}
\noindent As an example, suppose that we want to hedge a stock option with payoff $H = g(S_T)$ using the stock as hedging instrument.
In that case we get from  Proposition~\ref{prop:hedging} that
$$
\theta^H_t = \frac{ \nicefrac{d\langle \Pi^H, \widetilde{G}^{\text{stock}}\rangle } { db } \,(t)}{\nicefrac{d\langle \widetilde{G}^{\text{stock}}\rangle}{ db }\,(t)}\,.
$$

\paragraph{Computation of $\theta^*\!$.} The crucial task  in applying Proposition~\ref{prop:hedging} is to compute the instantaneous quadratic variations $ d\langle \widetilde{\Pi}^H, \widetilde{G}\rangle / db \,(t)$, and we now explain how  this can be
achieved for the claims considered in Subsection~\ref{subsec:derivatives}.  If $H$
represents a non-traded basic debt security,  an argument analogous to the proof of Theorem~\ref{thm:asset-price-dynamics} gives the representation of $\widetilde{\Pi}^H$ as stochastic integral with respect to  the martingales $M^Z$, $M^Y$ and $\mu^D - \gamma^{D, \bbF^\bbM}(dy,dt)$, and $\nicefrac{ d\langle \widetilde{\Pi}^H, \widetilde{G}\rangle}{db} \,(t)$ can be read off from this representation.

Next we turn to  the case where $H$ is  an option on  a traded assets with payoff $g(\Pi_T^1,\dots, \Pi_T^\ell)$ and we assume for simplicity that $g(0) =0$.   In order to compute  $\nicefrac{ d\langle \widetilde{\Pi}^H, \widetilde{G}\rangle}{db} \,(t)$  we need to find the martingale representation of $\widetilde{\Pi}^H$ with respect to  $M^Z$, $M^Y$ and $\mu^D - \gamma^{D, \bbF^\bbM}(dy, dt)$. Standard arguments can be used to show that such a representation exists, see for instance the proof of Lemma~3.2 in \citeasnoun{bib:frey-schmidt-12}.  However, identifying the integrands is more difficult.
A possible approach is to use the  It\^o formula for SPDEs   from \citeasnoun{bib:krylov-13}, see  Appendix~A for details.


\paragraph{Risk-minimizing strategies via regression.}
In order to circumvent the problem of finding  the martingale representation of $\widetilde{\Pi}^H$ one may use strategies with fixed discrete rebalancing dates and apply the results of \citeasnoun{bib:foellmer-schweizer-89}; this is sufficient for most practical purposes.
Consider a fixed set of trading  dates $0=t_0 < t_1 < \dots < t_m =T $. The  space  of \emph{admissible discrete trading strategies} consists  of all strategies $\phi^{(m)} = (\theta^{(m)}, \eta^{(m)})$ with $\theta_t^{(m)} = \sum_{j=0}^{m-1} \theta_j 1_{(t_j, t_{j+1}]} (t)$ and $\eta_t^{(m)}  = \sum_{j=1}^{m-1} \eta_j 1_{[t_j, t_{j+1})} (t) + \eta_m \ind{t=T} $ such that $\theta_j $ and  $ \eta_j$ are $\F_{t_j}^\bbM$-measurable. Moreover,  for all $0 \le j \le m-1$ the random variable $   \theta_j^\top (\widetilde{G}_{t_{j+1}} -  \widetilde{G}_{t_{j}} )$ is square integrable. Note that $\theta^{(m)}$ is left continuous and that $\eta^{(m)}$ is right continuous.
\citeasnoun{bib:foellmer-schweizer-89} show that the strategy  $(\phi^{(m)})^*$ that minimizes the remaining risk over all admissible discrete trading  strategies with terminal value $V_T = H$ can be described as follows: For $0 \le j \le m-1$ the random vector $\theta_j^*$,  is determined from the regression equation
$$ \widetilde{\Pi}_{t_{j+1}}^H - \widetilde{\Pi}_{t_j}^H = \sum_{i=1}^\ell (\theta_j^*)^\top  \big(\widetilde{G}_{t_{j+1}}-\widetilde{G}_{t_j}\big) + \epsilon_{j+1}
$$
where $E(\epsilon_{j+1}\mid \F_{t_j}^\bbM) =0 $ and where  $E\big ( \epsilon_{j+1}  (\widetilde{G}_{t_{j+1}}-\widetilde{G}_{t_j}) \mid \F_{t-1}^\bbM\big ) = 0$.  The cash position is given by
$\eta*_{j+1} = \widetilde \Pi_{t_j} - (\theta_j^*)^\top \widetilde{\Pi}_t$ so that $V_{t_j}^{\phi^{(m)}} = \Pi_{t_j}^H$ for all $j$.
In order to compute $(\theta_j)^*$ one may therefore generate realisations of $\widetilde{\Pi}_{t_{j+1}}^H - \widetilde{\Pi}_{t_j}^H$ and of $\widetilde{G}_{t_{j+1}}  -\widetilde{G}_{t_j}$  via Monte Carlo; $\theta_j^*$  can then be computed from these  simulated data via standard regression methods.

\paragraph{Further comments.} Note that the hedging strategies for options on traded assets can be expressed as functions of the current filter density $\pi(t)$.
In the case where the asset value jumps downward at the dividend dates, that is for $\kappa =1$, the model is inevitably incomplete. For $\kappa =0$ it is possible to give conditions that ensure that the market is complete: loosely speaking, the number of traded risky securities  must be equal to $l +1$, where $l$ is the dimension of the process $Z$.
For details on both issues we refer to Appendix~A.

\subsection{Calibration of the filter density} \label{subsec:calibration}

In our setup pricing formulas and hedging strategies depend on the current filter density $\pi(t)$.    Hence an investor who wants to use the model needs to  estimate of $\pi(t)$ from prices of traded  securities at time $t$. In this section w explain how this can be achieved  by means of a quadratic optimization problem with linear constraints. We assume  that
a Galerkin approximation of the form $\pi^{(m)}(t) = \sum_{i=1}^m \psi_i e_i$ with smooth basis functions $e_1, \dots, e_m$ is used to approximate the filter density $\pi(t)$ and that  we observe prices $
\Pi_1^*, \dots ,\Pi_\ell^*$ of $\ell$ traded securities with full information value $h_j(t,v)$, $1\le j \le \ell$. In order to match the observed prices perfectly, the vector  of Fourier coefficients $\vecb{\psi} = (\psi_1, \dots, \psi_m)^\prime$ needs to satisfy the following $\ell +1$ linear  constraints
\begin{align} \label{eq:constraints-for-calibration}
\sum_{i=1}^m \psi_i (e_i,1)_{S^X} & = 1  \text{ and }\,
\sum_{i=1}^m \psi_i \big(e_i,h_j(t,\cdot)\big)_{S^X} = \Pi_j^*, \; 1 \le j \le \ell\,;
\end{align}
moreover, it should hold that $\vecb{\psi} \ge 0$ in order to prevent  $\pi^{(m)}(t)$ from becoming negative.   Typically, $m>\ell$ so that the constraints \eqref{eq:constraints-for-calibration} do not determine the Fourier coefficients uniquely. In that case one needs to apply a regularisation procedure. Following \citeasnoun{bib:hull-white-06} who face a similar issue in the calibration of the implied copula model to CDO tranche spreads,  we  propose to minimize the $L^2$-norm of the second derivative of $\pi^m(t,\cdot) $ over all nonnegative $\vecb{\psi}$ that satisfy the constraints \eqref{eq:constraints-for-calibration}; this produces a maximally smooth initial density.

Denote by $e_i^{\prime\prime}$ the second derivative of $e_i$ and define the symmetric and positive definite matrix $\Xi$ by $\Xi_{ij} = \big( e_i^{\prime\prime} , e_j^{\prime\prime} \big )_{S^X}$. Since
$$ \int_{S^X} \Big( \frac{d^2 \pi^{(m)}(t,x)}{dx^2} \Big )^2\, dx = \sum_{i, j =1}^m \psi_i \psi_j \big( e_i^{\prime\prime} , e_j^{\prime\prime} \big )_{S^X} = \vecb{\psi}^\prime \Xi \vecb{\psi},
$$
minimization of the $L^2$-norm of $\frac{d^2}{dx^2} \pi^{(m)}(t,x)$ thus leads to the quadratic optimization problem
$$
\min_{\vecb{\psi} \ge 0} \vecb{\psi}^\prime \Xi \vecb{\psi} \text{ such that $\vecb{\psi}$ satisfies \eqref{eq:constraints-for-calibration}.}$$
This problem can be solved with standard optimization software; a numerical example is discussed in Section~\ref{sec:simulations}.

For a full calibration of the model one needs to determine also the volatility $\sigma$ of $V$ and (parameters of) the function $a$. A natural approach is to determine these parameters by calibration to observed option prices; details are left for future research.

\section{Numerical Experiments}\label{sec:simulations}

In this section we illustrate the model with a number of  numerical experiments. We are particularly interested in the  asset price dynamics under incomplete information.
We use the following setup for our analysis: Dividends are paid annually; the  dividend size is  modelled as   $d_n = \delta_n (V_{t_n}-K)^+$ where $\delta_n$  is Beta-distributed with  mean equal to $2\%$ and standard deviation equal to $1.7\%$.  The process  $Z$ is two-dimensional with $a_1(v) = c_1 \ln v$ and $a_2(v) = c_2 \big (\ln K + \sigma - \ln v)^+$ \footnote{We smooth $a_2$ around the kink at $\ln v = \ln K + \sigma$; details do not matter.};  for $c_2 >0$ this  choice of $a_2$ models the idea that prices are very informative as as soon as the asset value is less than one standard deviation away from the  default threshold, perhaps because the firm is monitored particularly closely in that case. The remaining parameters are given in  in Table~\ref{table:parameters}.

\begin{table}
\begin{center}{\small
\begin{tabular}{lcccc}
$K$ & $r$ &     $\sigma$ (vol of GBM) & $\kappa$ & { initial filter distribution $\pi_0$}\\ \hline
20  & 0.02&         0.2 & 1 & {  $V- K \sim LN(\ln15, 0.2) $ }\\
\end{tabular}}
\end{center}
\caption{\label{table:parameters} Parameters used in simulation study.}
\end{table}

In order to generate a trajectory of the filter density $\pi(t)$ with initial value $\pi_0$ and related quantities such as the stock price $S_t$ we proceed according to the following steps.

\begin{enumerate}
\item Generate a random variable $V \sim \pi_0$,  a trajectory $(V_s )_{s=0}^T$  of the asset value process with  $V_0 =V$ and the   associated trajectory $(Y_s)_{s=0}^T$ of the default indicator process.
\item Generate   realizations $(D_s )_{s=0}^T$ and $(Z_s)_{s=0}^T$,  using the  trajectory $(V_s )_{s=0}^T$  generated in Step~1 as input.
\item Compute  for the observation generated in Step~2  a trajectory $(u(s))_{s=0}^T$ of the unnormalized filter density with initial value $u(0) = \pi_0$, using the Galerkin approximation described in Section~\ref{subsec:finite-dim-approx}. Return $\pi(s) = (1-Y_s) \big( u(s) / ( u (s), 1)_{S^X}\big)$ and
    $ S_s = (1-Y_s) (\pi(s), h^\text{stock})_{S^X}$, $ 0\le s    \le T$.
\end{enumerate}
For details on the numerical methodology including the choice of the basis functions, numerical methods to solve the SDE system~\eqref{eq:SDE-for-Psi} arising  from the Galerkin method and  tests for the accuracy of the numerical implementation  we refer to Chapter~4 of \citeasnoun{bib:roesler-16}.

Next we describe the results of our numerical experiments.

In Figure~\ref{fig:stock-price-only-dividends}  we plot a trajectory of the stock price ${S}$ and of the corresponding full information value $h^{\text{stock}}(V_t)$ for the case where the modelling filtration consists only of the dividend information ($c_1 = c_2=0$). This can be viewed as an example of the discrete noisy accounting information considered in \citeasnoun{bib:duffie-lando-01}.  We see that $S$ has very unusual  dynamics; in particular  it evolves deterministically between  dividend dates.

\begin{figure}[h]
\begin{center}
\includegraphics[width=13.5cm,height=7.0cm]{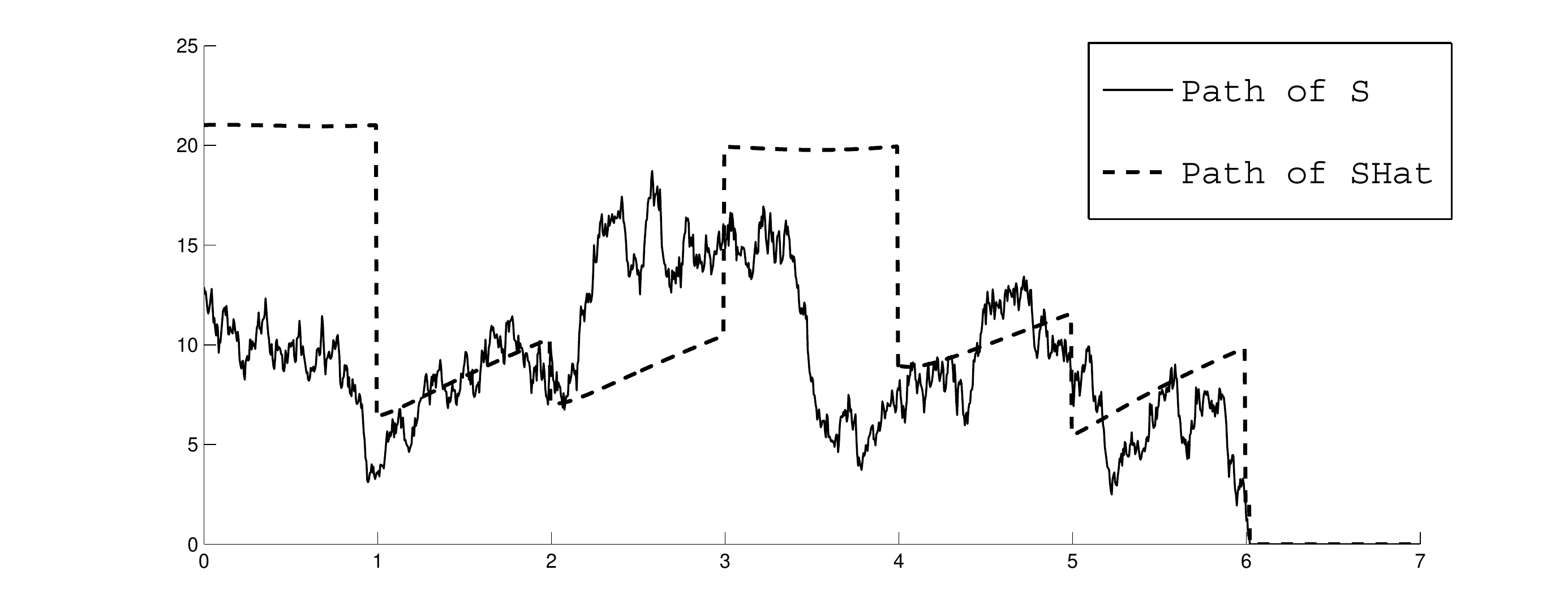}
\end{center}
\caption{\small A simulated path of the full information value $h^{\text{stock}}(V_t)$ of the stock (dashed line) and  of the stock price ${S}$ (normal line, label SHat) for $c_1=c_2=0$ (only dividend information). \label{fig:stock-price-only-dividends}}
\end{figure}

Next we  show that  more realistic asset price dynamics can be obtained  by adding the filtration $\bbF^Z$ to the  modelling filtration.  In Figure~\ref{fig:stock-price-c1-4-c2-0} we plot a typical stock price trajectory  together with the  full information value $h^{\text{stock}}(V_t)$ for  the parameter values $c_1 =4$ and $c_2 =0$.
Clearly,  $S_t$ has nonzero volatility between dividend dates.   A comparison of the two trajectories moreover shows  that the stock price  jumps to zero at the default time $\tau$; this reflects the fact that the default time has an intensity under incomplete information so that default comes as a surprise. The corresponding filter density $\pi(t)$ is plotted in Ficure~\ref{fig:evol-of-pi}

For comparison purposes we finally consider the parameter set  $c_1=4, c_2 =25$. For these parameter values  default is ``almost predictable" and the model behaves similar to a structural model.
This can be seen from Figure~\ref{fig:defint-comparison} where we plot the default intensity for both parameter sets. Note that for
$c_1=4, c_2 =25$,  the default intensity  is close to zero most of the time and very large immediately prior to default (in fact almost twice as large as in the case where $c_2 =0$.).

In Figure~\ref{fig:calibration-lehman}  we finally  present the result of a small calibration exercise,  where $\pi(t)$ was calibrated to five-year CDS spreads of Lehman brothers using the methodology described in the previous section. The data range over the period September 2006 to September 2008 (Lehman filed for bankruptcy protection  on September 15, 2008). Since under full information CDS spreads are homogeneous of degree zero in $V$ and $K$, we took the default threshold  equal to $K =1$ so that the  numbers on the $x$-axis can be viewed as ratio of asset over liabilities. It can be seen clearly that  prior to default the mass of $\pi(t)$ is concentrated close to  the default threshold.

\begin{figure}
\begin{center}
\includegraphics[width=13.5cm,height=7.0cm]{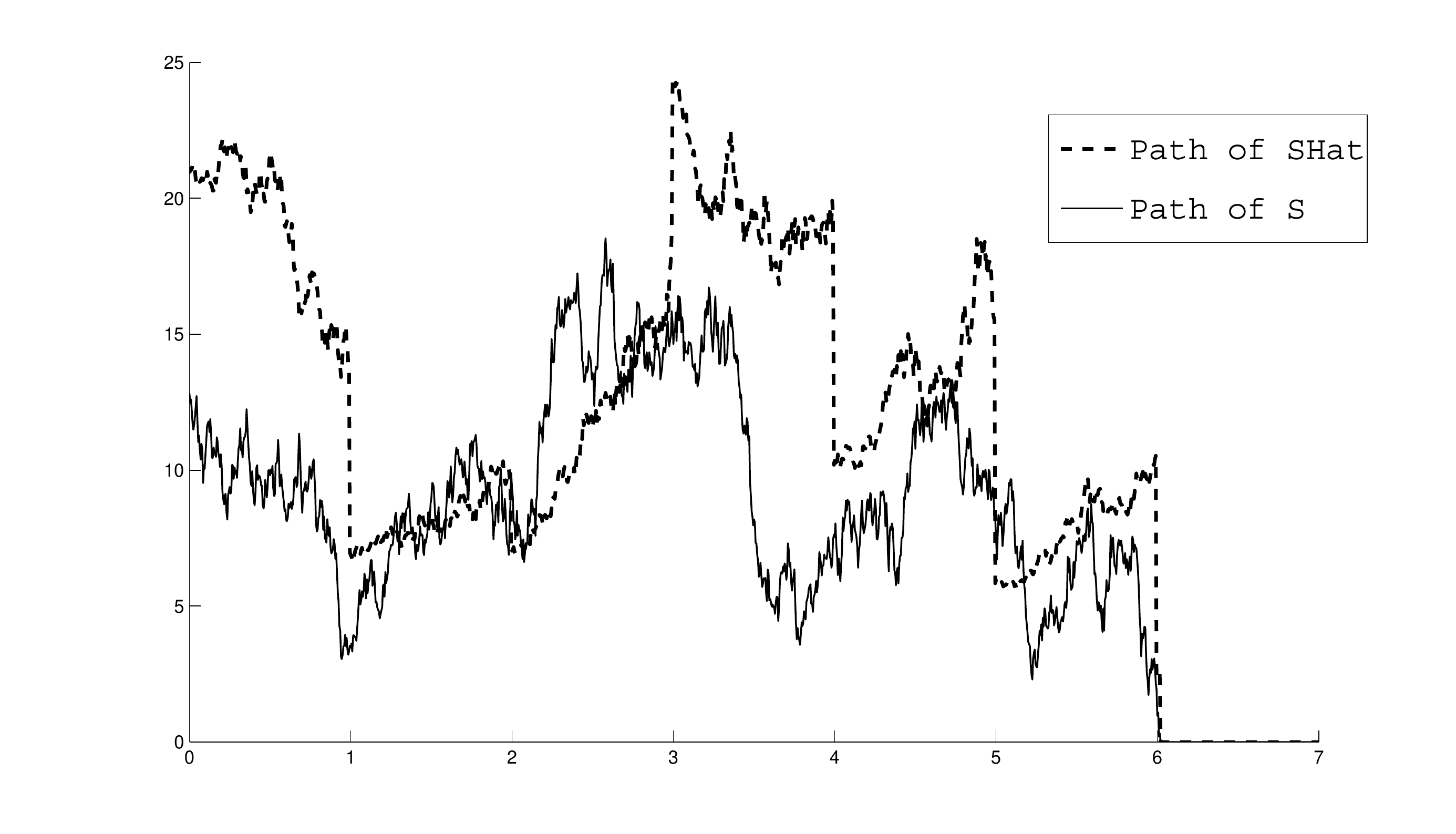}
\end{center}
\caption{\small A simulated path of the full information value $h^{\text{stock}}(V_t)$ of the stock (dashed line) and  of the stock price ${S}_t $ (normal line, label Shat) for $c_1= 4, c_2=0$. \label{fig:stock-price-c1-4-c2-0}}
\end{figure}

\begin{figure}
\begin{center}
\includegraphics[width=13.5cm,height=7.5cm]{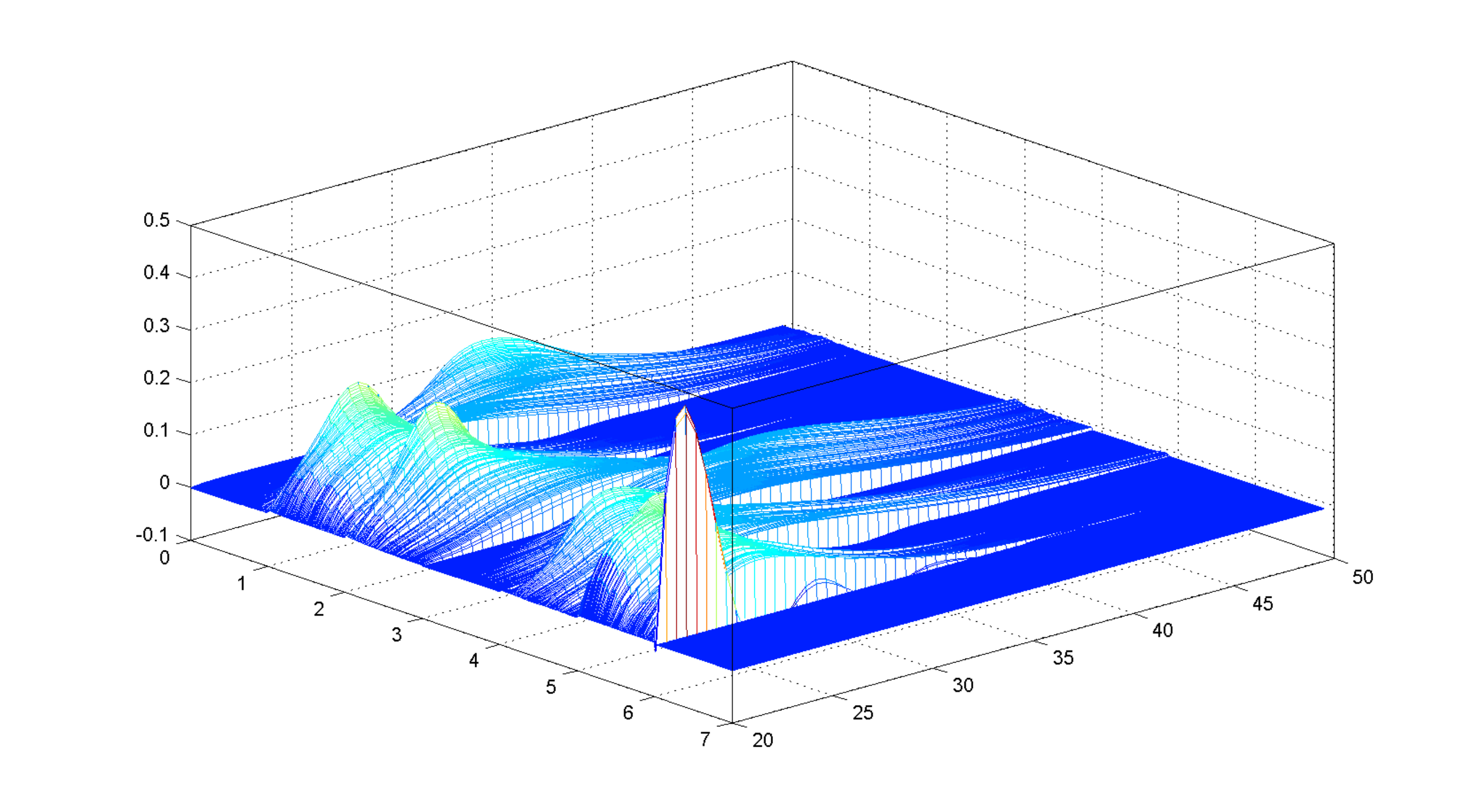}
\end{center}
\caption{\label{fig:evol-of-pi} \small A simulated realisation of the conditional density $\pi(t)$. Note how the mass of $\pi(t)$ concentrates at the default threshold shortly before $\tau$. }
\end{figure}


\begin{figure}
\begin{center}
\includegraphics[width=13.5cm,height=7.0cm]{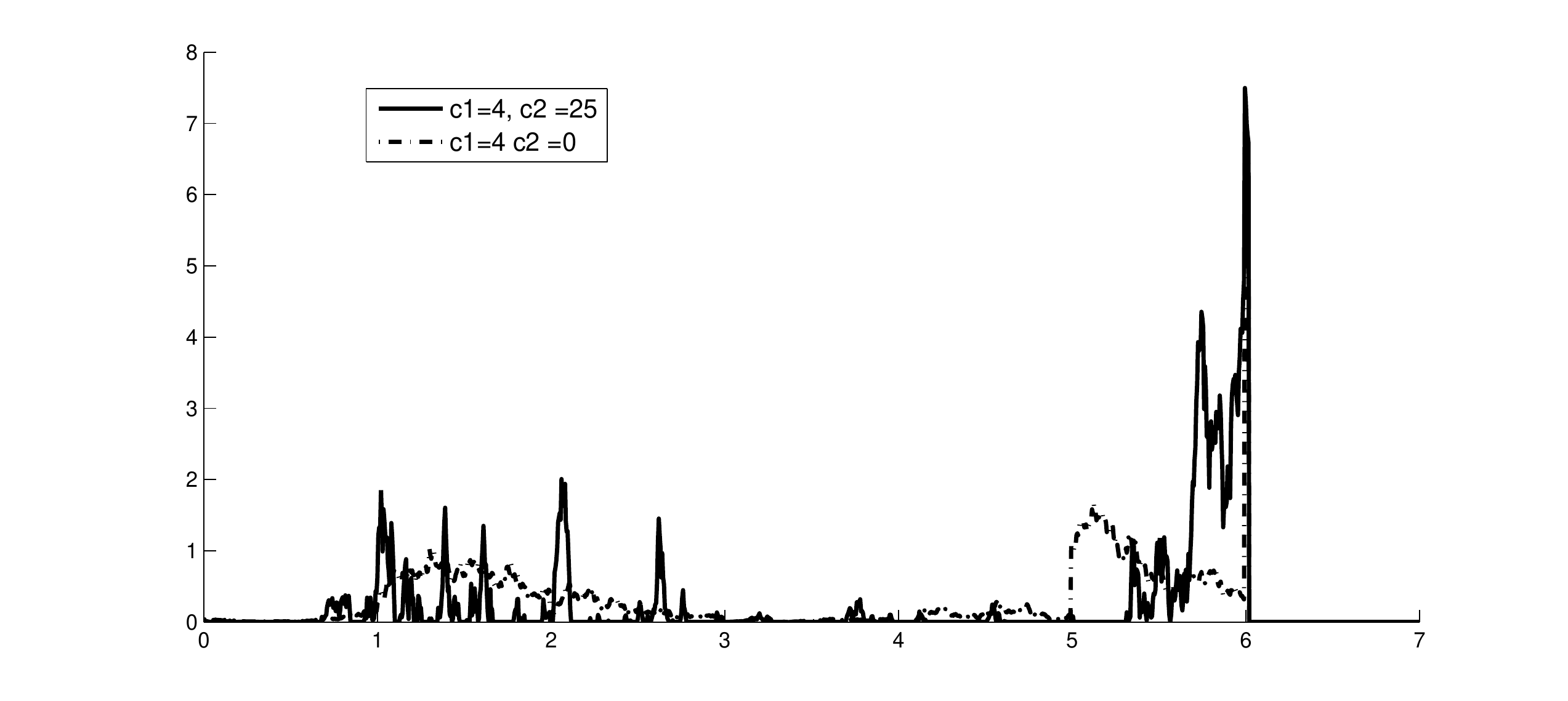}
\end{center}
\caption{\small A simulated path of the default intensity for for $c_1= 4, c_2=0$ (dashed) and for $c_1= 4, c_2=25$ (straight line).\label{fig:defint-comparison}}
\end{figure}
\begin{figure}
\begin{center}
\includegraphics[width=13.0cm,height=8.0cm]{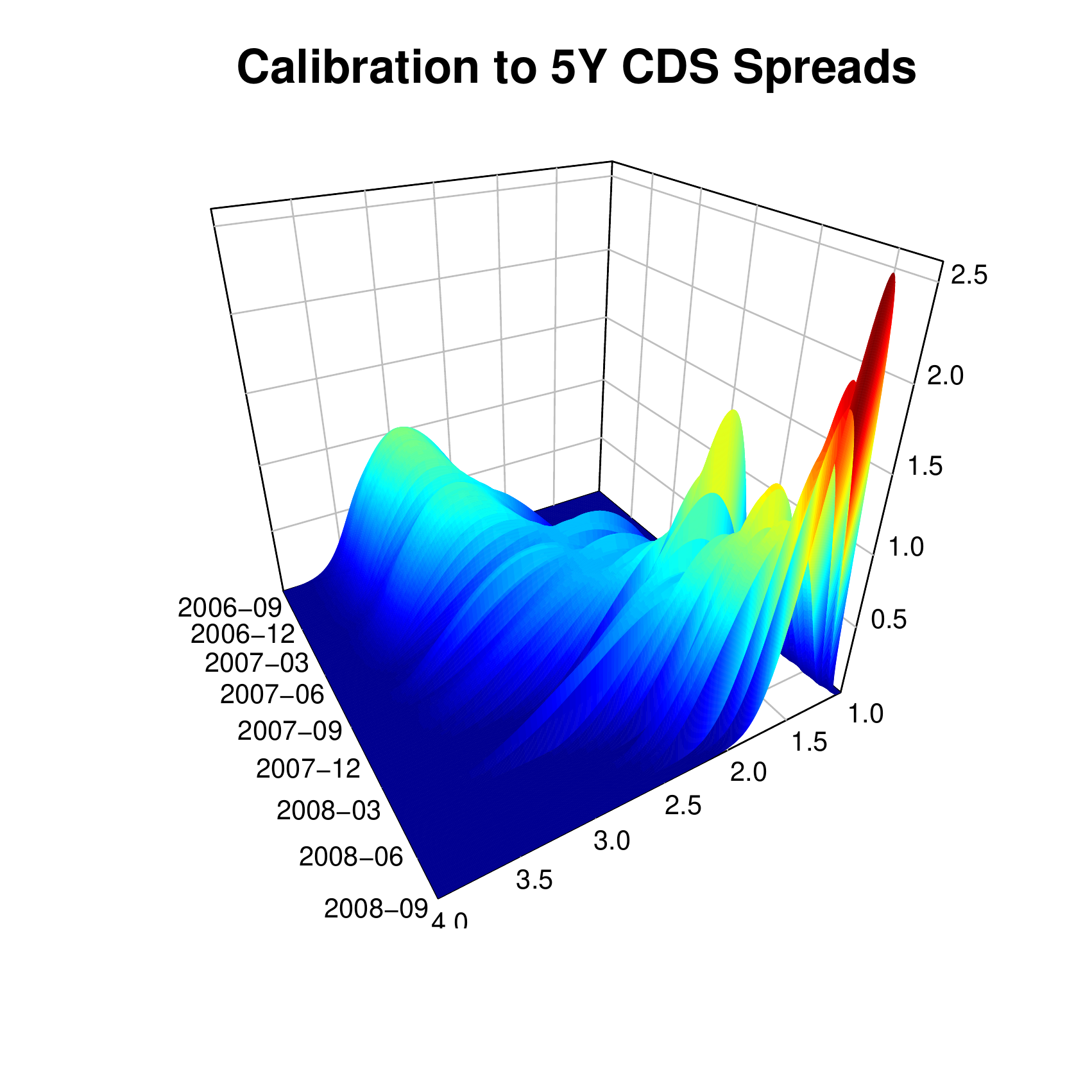}
\end{center}
\caption{\small Result of a calibration of the filter density to 5-year CDS-spreads of Lehman brothers prior to default of the bank.  In this example the default threshold is $K=1$. \label{fig:calibration-lehman}}
\end{figure}

\section{Outlook and Conclusion}

This paper has developed a theory of derivative asset analysis for structural credit risk models under incomplete information using stochastic filtering techniques. In particular we managed to derive the dynamics of traded securities which enabled us to study the pricing and the hedging of derivatives. To conclude we briefly mention a couple of financial problems where this theory could prove useful.

To begin with it might be  interesting to study  contingent convertible bonds, also known as CoCos, in our setup. A CoCo  is a convertible   bond that is automatically  triggered once the issuing company  (typically a financial institution) enters into financial distress. At the trigger event the bond is either converted into  equity or into an immediate cash-payment that is substantially lower than the nominal value of the bond.  Modelling the trigger mechanism adequately is a crucial part in the analysis of CoCos. The CoCos that have been issued so far have a so-called accounting trigger based on capital adequacy ratios.  It is difficult to include this directly into a formal pricing model; many pricing approaches  therefore model the  conversion time $\tau^{\text{CoCo}}$ as a first passage time of the form $\tau^{\text{CoCo}} = \inf \{t \in \mathcal{T} \colon V_t \le K^\text{CoCo}\}$ for a conversion threshold $K^\text{CoCo} >K $ and a set of monitoring dates $\mathcal{T} \subset [0, \infty)$.   This valuation approach is however  difficult to apply in practice,  since investors are not able to track the asset value continuously in time, see for instance \citeasnoun{bib:spiegeleer-schoutens-12}.     Our setup is where $V$ is not fully observable is well-suited for dealing with this issue in a consistent manner. First results in this direction can be found in Chapter~5 of~\citeasnoun{bib:roesler-16}.

Our framework could also be used to study  derivative asset analysis for sovereign bonds. Several fairly recent papers  have proposed structural  models with endogenous default for sovereign credit risk, see for instance~\citeasnoun{bib:andrade-09} or \citeasnoun{bib:mayer-13}. Roughly speaking, in these models   default is given by a first passage time,
$$ \tau = \inf\{ t \ge 0 \colon \tilde{V}_t \ge  \tilde{K}_t \} = \inf\{ t \ge 0 \colon V_t := \tilde{V}_t/ \tilde{K}_t \le 1\},  $$
where the process $\tilde V$  is a measure of the expected future economic performance of the sovereign and where  the   threshold process $\tilde{K}$  is  chosen   by the sovereign in an attempt to  balance the benefits accruing from lower debt services against  the adverse economic implications of a default such as reduced access to capital markets. It is reasonable to assume that $\tilde V$ and $\tilde K$ are not fully observable for outside investors, for instance because it is hard to predict the outcome of the sovereign's decision process in detail. Hence one is led to  a model of the form \eqref{eq:def-tau} with ``asset value $V = \tilde{V}/ \tilde{K}$ and default threshold $K=1$. The results of the present paper can be used to  derive the dynamics of sovereign credit spreads in this setup; this is  important for the pricing of options on sovereign bonds and for risk management purposes.

\appendix
\section{Additional Results} \label{app:results}

\subsection{Filter equations.} In the next corollary  we state  the  filtering equations for $\bbF^\bbM$.
\begin{corollary} \label{thm:filter-equations}
For $f \in C^{1,2}([0,T] \times S^X)$ the optional projection $\widehat{f}_t$
has dynamics
\begin{equation} \label{eq:filter-equation}
\begin{split}
\widehat{f}_t &= \widehat{f}_0 + \int_0^t \Big(\widehat{  \frac{df}{dt}}\Big)_s + (1-Y_{s-})
(\widehat{ \mathcal{L}_X f})_s \, ds + \int_0^{t \wedge \tau} (\widehat{f a})_s^\top -
\widehat{f}_s \widehat {a}_s^{\top} \, d M_s^{Z,\bbF^\bbM}
 \\& + \int_0^{t \wedge \tau} (f(s,K) - \widehat{f}_{s-}) \,d  M_s^Y  +
\int_0^{t \wedge \tau} \int_{\R^+} \Big(
\frac{( \widehat{f(\cdot- \kappa y) \varphi(y,\cdot)})_{s-}
}{ (\widehat{\varphi(y)})_{s-}} - \widehat{f}_{s-}\Big ) \; \mu^D(dy,ds)\,.
\end{split}
\end{equation}
\end{corollary}
\begin{proof}
Recall that $\widehat{f}_t = \ind{\tau \le t} f(t,K) + \ind{\tau >t} \pi_t f.$
Hence it holds that
\begin{equation*}
d \widehat{f}_t = Y_{t-} \frac{df}{dt}(t,K) dt + \left( f(t,K) - \pi_{t} f\right) d Y_t + (1- Y_{t-}) d \pi_t f
\end{equation*}
Substituting the dynamics of $\pi_t f$ in this equation gives~\eqref{eq:filter-equation}.
\end{proof}
Alternatively, one can derive the filter equations using the innovations approach to nonlinear filtering. For this one has to show first that every $\bbF^\bbM$ martingale can be represented as a sum of stochastic integrals with respect to  the processes $Y_t - \Lambda_{t \wedge \tau}$ and     $M^{Z,\bbF^\bbM}_t$,    and with respect to  the   random measure  measure $\mu^D -\gamma^{D,\bbF^{\bbM}}$. Standard arguments can then be used to identify the integrands in the martingale representation of $\widehat{f}_t - \int_0^t (\widehat{\mathcal{L}_X f})_s ds$. This is the route taken in \citeasnoun{bib:cetin-12} for the case without dividend payments.

\subsection{Dynamic hedging}
Next we give a few additional results related to our analysis of dynamic hedging in Section~\ref{subsec:hedging}.

\paragraph{Computation of hedging strategies.} We now explain how  the  It\^o formula for SPDEs can be used to compute the integrands in the martingale representation of an option on traded assets; this is needed during the  application of   Proposition~\ref{prop:hedging}.

We know from Section~\ref{subsec:derivatives} that $\widetilde{\Pi}_t^H = \widetilde{C}^H(t, \pi(t))$
Using similar arguments as in the proof of Proposition~\ref{lem:d-pit-f}, the dynamics of the conditional density $\pi(t)$  can be derived from the dynamics of $u(t)$ given in \eqref{eq:du-t-with-jumps}: for $t< \tau$ it holds that
\begin{equation*}
\begin{split}
d \pi (t)  &=  \big(\mathcal{A}^* \pi(t) +  \pi(t)  \lambda_t \big ) dt
 +     \pi(t) \big( a^\top  -   \widehat{a}^\top_t\big )   dM_t^{Z}
+ \int_{\R^+} \!\! \frac{ S_{\kappa y} \big (\pi(t-) \varphi(y,\cdot) \big )}{ (\widehat{\varphi (y)})_{t-}} - \pi(t-) \,  \mu^D (dy,dt)\,.
\end{split}
\end{equation*}
Denote for $v \in H^1_0 (S^X)$ by $\widetilde{C}^H_ {[v]} (t, \pi) $ the directional derivative of ${\widetilde C}^H$ in direction $v$, that is
$$ \widetilde{C}^H_{[v]} (t, \pi) = \frac{d}{ds} {\Big \vert}_{s=0} \widetilde{C}^H (t, \pi + sv)  \,.$$
Suppose that $\widetilde{C}^H$ satisfies the regularity conditions of~\citeasnoun{bib:krylov-13} (essentially this means that the  first and second directional derivative of $\widetilde{C}^H$ exists in every point $\pi \in H^1_0 (S^X)$).
 Then Theorem 3.1 of \citeasnoun{bib:krylov-13}) gives the following martingale representation for the discounted option price:
\begin{align} \nonumber
\widetilde{C}^H(t, \pi(t)) &= \widetilde{C}^H(0, \pi(0)) +  \sum_{i=1}^l \int_0^{t \wedge \tau} \!\! \widetilde{C}^H_{[ \pi(s)(a^i-(\widehat{a^i})_s)]} (s, \pi(s))\, dM_{s,i}^{Z}
\label{eq:martingale-representation-2}
 - \int_0^{t \wedge \tau} \!\! \widetilde{C}^H(s, \pi(s))\, d M^Y_s \\
+& \int_0^{t \wedge \tau} \!\! \int_{\R^+} \widetilde{C}^H \Big ( s, \frac{ S_{\kappa y} \big (\pi(s-) \varphi(y,\cdot) \big )}{ (\widehat{\varphi (y)})_{s-}} \Big ) - \widetilde{C}^H\big ( s, \pi(s-) \big )\, (\mu^D - \gamma^{D, \bbF^\bbM})(dy, ds)\,.
\end{align}
In \eqref{eq:martingale-representation-2} the It\^o formula for SPDEs  is used to determine the integrands with respect to $dM_{s,i}^{Z}$;   the  integrands with respect to $M^Y$ and with respect to $(\mu^D - \gamma^{D, \bbF^\bbM})(dy, dt)$ can be determined by elementary arguments. All  integrands  (and hence the risk-minimizing hedging strategy) are functions of the current filter density $\pi(t)$; the practical computation of the directional derivative of $\widetilde{C}^H$ in the integral with respect to $M^Z$  (and of the other integrands) can be done with Monte Carlo.  Note that we have not established the regularity of $\widetilde{C}$ required in \citeasnoun{bib:krylov-13}; this very technical issue is  left for future research.

\begin{example} As a toy example we consider the problem of hedging an option with payoff $H = g(S_T)$ using the stock as hedging instrument. To simplify the notation we consider the case where $Z$ is a one-dimensional process. First, we get from  Proposition~\ref{prop:hedging} that
$$
\theta^H_t = \frac{ \nicefrac{d\langle \Pi^H, \widetilde{G}^{\text{stock}}\rangle } { db } \,(t)}{\nicefrac{d\langle \widetilde{G}^{\text{stock}}\rangle}{ db }\,(t)}\,.
$$
Between dividend dates, that is for $t \in [0,T] \setminus \mathcal{T}^D$ this gives
$$
\theta^H_t = \frac{\big (( {\widehat{h^\text{stock}} a})_{t} - \widetilde{S}_{t} \widehat{a}_{t} \big )
\widetilde{C}^H_{[ \pi(t)(a-\widehat{a}_{t})]} (t, \pi(t)) + \lambda_t \widetilde{S}_{t-} \widetilde{C}^H(t, \pi(t-)
}
{\big( ( {\widehat{h^\text{stock}} a})_{t}  - \widetilde{S}_{t-} \widehat{a}_{t-} \big )^2 + \lambda_t^2 \widetilde{S}_{t-}^2}\, ;
$$
at the dividend date $t_n$ we get
$$
\theta^H_t = \frac{ \displaystyle{\int_{\R^+}}
 \Big ( \widetilde{C} \Big ( t, \frac{ S_{\kappa y} \big (\pi(t-) \varphi(y,\cdot) \big )}{ (\widehat{\varphi (y)})_{t-}} \Big ) - \widetilde{C}\big ( t, \pi(t-)) \Big ) \Big ( y +  \frac{( \widehat{h^{\text{stock}} \varphi(y)})_{t-} }{ (\widehat{\varphi (y)})_{t-}}- S_{t-} \Big) \,(\widehat{\varphi}(y))_{t-} dy }
{\displaystyle{\int_{\R^+}} \Big ( y +  \frac{( \widehat{h^{\text{stock}} \varphi(y)})_{t-} }{ (\widehat{\varphi (y)})_{t-}}- S_{t-} \Big)^2  \,(\widehat{\varphi(y)})_{t-} dy }\,.
$$
\end{example}

\paragraph{Market completeness.} Finally we discuss market completeness in our setup.  For this we consider a variant of the model without dividend payments.  This assumption is essential; with dividends the market is generically incomplete.
Consider an option on traded assets with maturity $T$ and payoff $H$. Then the discounted price process $\widetilde{\Pi}^H$ has a martingale representation of the form
\begin{equation}\label{eq:martingale-repres-PiH}
\widetilde{\Pi}^H_t = \widetilde{\Pi}^H_0 + \sum_{j=1}^l \int_{0}^{t \wedge \tau } \xi_s^{M_j^Z,H} d M_{s,j}^{Z} +
  \int_{0}^{t \wedge \tau }\xi_s^{Y,H}d M_s^Y\,.
\end{equation}
As shown in Theorem~\ref{thm:asset-price-dynamics}, a  similar representation holds for the discounted gains processes of the traded assets; the integrands are denoted by $\xi_{i}^{M_j^Z}$ and  $\xi_{i}^{Y}$, $1 \le i \le \ell, 1 \le j \le l$.  A  perfect hedging strategy  $\theta^*$ satisfies for all $t \le T$ the relation $\widetilde{\Pi}^H_t = \widetilde{\Pi}^H_0 + \sum_{i=1}^\ell \int_0^{t} \theta^*_{s,i} d \widetilde{G}_{s,i}$;
the cash position is determined from the selffinancing condition. Now we get that
\begin{equation}\label{eq:perfect-hedging}
\sum_{i=1}^\ell \int_0^{t} \theta^*_{s,i} d \widetilde{G}_{s,i}
= \int_{0}^{t \wedge \tau } \sum_{j=1}^l \Big ( \sum_{i=1}^\ell \theta^*_{s,i} \xi_{s,i}^{M_j^Z}\big )  dM_{s,j}^Z +
 \int_{0}^{t \wedge \tau } \sum_{i=1}^\ell \theta^*_{s,i} \xi_{s,i}^{Y} d M_s^Y\,.
\end{equation}
Comparing \eqref{eq:martingale-repres-PiH} and \eqref{eq:perfect-hedging} we see that for $t \le \tau$ a perfect replication strategy  $\theta^*_t$ has to solve the following  $l+1$ dimensional system of linear equations
\begin{equation}\label{eq:eqn-for-theta}
\sum_{i=1}^\ell \theta^*_{t,i} \xi_{t,i}^{M_j^Z} = \xi_t^{M_j^Z,H} , \; 1 \le j \le j\, ,   \, \text{ and } \sum_{i=\ell}^l \theta^*_{t,i} \xi_{t,i}^{M^Y} = \xi_t^{M^Y,H} .
\end{equation}
Modulo integrability conditions, for every option on traded assets a perfect hedging therefore  exists  if and only  if  the system \eqref{eq:eqn-for-theta} has a solution for all $t \le \tau$ and  every right hand side $(\xi_t^{M_1^Z,H},\dots, \xi_t^{M_l^Z,H},\xi_t^{M^Y,H})^\top$. Loosely speaking, for  complete markets one thus needs to have for every $t< \tau $ at least $l+1$ locally independent  traded assets. With dividend payments we would get  an additional equation for every $y$ in the support of $\gamma^{D, \bbF^\bbM}(dy, \{t_n\})$ so that at $t_n$  \eqref{eq:eqn-for-theta}  becomes a system with infinitely many equations for finitely many unknowns and hence generically unsolvable.

\section{Proofs}\label{app:proofs}

{\small
\begin{proof}[Proof of Lemma~\ref{lemma:finite-stock-price}]
Define  the gains process $G_t = V_t + D_t$, the discounted gains process
\begin{equation} \label{eq:def-tildeG}
\widetilde G_t = e^{-r t} V_t + \int_0^t   e^{-r s} d D_s =: \widetilde V_t + \widetilde D_t
\end{equation}
and the \emph{cum-dividend asset value process} $\bar V_t $ with  $\bar V_t = V_0 + \int_0^t r \bar V_s ds + \int_0^t  \sigma \bar V _s d B_s$.  We first show  that $\widetilde G$ is a martingale. In fact, by \eqref{eq:dVt} we have
$$ d \widetilde G_t = \sigma \widetilde V_t d B_t -  e^{-r t} d D_t +  e^{-r t} d D_t = \sigma \widetilde V_t d B_t \,;
$$
so that $\widetilde G$ is a local martingale. Moreover, $E^Q(\widetilde V_t^2) \le E^Q\big((e^{-r t} \bar V_t)^2\big) = e^{\sigma^2 t} E^Q(V_0^2)$. Hence
 we get
$$
E^Q ( [ \widetilde G ]_T ) = E^Q \Big (\int_0^T \sigma^2 \widetilde{V_s}^2 ds \Big) \le
E^Q \Big (\int_0^T \sigma^2 \big (e^{-r s} \bar V_s \big)^2  ds \Big) < \infty\,,
$$
and $\widetilde G$ is a square integrable  true martingale.
Now it obviously holds that $\sum_{t_n \ge 0} e^{-r t_n}   d_n  =  \lim_{n \to \infty} \widetilde D_{t_n} $.
Moreover, as $\widetilde G_{t_n} = \widetilde V_{t_n} + \widetilde D_{t_n}  > \widetilde D_{t_n}$, one has
$E^Q (\widetilde D_{t_n} \mid V_0 ) \le E^Q ( \widetilde G_{t_n} \mid V_0 ) = V_0$. Since $\widetilde D_{t_n}$ is an increasing process we get from monotone integration that
\begin{equation}\label{eq:estimate-dividend-value}
E^Q \Big (\sum_{t_n \ge 0}   e^{-r t_{n}} d_n \mid V_0\Big)  =  E^Q \Big(\lim_{n \to \infty} \widetilde D_{t_n}  \mid V_0 \Big ) = \lim_{n \to \infty} E^Q\big (\widetilde D_{t_n} \mid V_0 \big) \le V_0\,.
\end{equation}
In order to show that one has equality in \eqref{eq:estimate-dividend-value} we have to show that $\lim_{n \to \infty} E^Q\big (\widetilde V_{t_n} \mid V_0\big) =0$.  We have the estimate
\begin{align}\label{eq:estimate-dividend-value2}
E^Q \Big (\sum_{t_n \ge 0}   e^{-r t_{n}} d_n \mid V_0\Big) \ge
E^Q \Big (\sum_{t_n \ge 0}   \delta_n \widetilde{V}_{t_n} - \delta^n e^{-r t_{n}} K  \mid V_0\Big) = E^Q(\delta_1) \Big ( \sum_{n \ge 0} E^Q\big (\widetilde V_{t_n} \mid V_0\big) - K \sum_{n \ge 0}e^{-r t_{n}} \Big).
\end{align}
Recall that the dividend dates are equidistant by Assumption~\ref{ass:assets-and-dividends}.2, so that $\sum_{n \ge 0}e^{-r t_{n}} < \infty$. Suppose now that there is some $c >0$ such that $E^Q\big (\widetilde V_{t_n} \mid V_0\big) >c $ for infinitely many $n$. Together with \eqref{eq:estimate-dividend-value2} this implies that $E^Q \Big (\sum_{t_n \ge 0}   e^{-r t_{n}} d_n \mid V_0\Big) = \infty$, contradicting the inequality  \eqref{eq:estimate-dividend-value}, and we conclude that $\lim_{n \to \infty} E^Q\big (\widetilde V_{t_n} \mid V_0\big) =0$.
\end{proof}

\begin{proof}[Proof of Identity \eqref{eq:different-filtration-1}]

This identity will follow from  the relation
\begin{equation} \label{eq:different-filtration-2}
 E^Q \big ( h(V_t) \ind{\tau >t} \mid \F_t^\bbM \big ) =
 E^Q \big ( h(V_t^\tau) \ind{\tau >t} \mid \F_t^{Z^\tau} \vee \F_t^Y \big ),
\end{equation}
where $\bbF^{Z^\tau}$ is the filtration generated by the stopped process $Z^\tau$.  To
this, note first that $\bbF^{Z^\tau} \vee \bbF^Y$ is a subfiltration of $\bbF^{\bbM} $ (as
$\tau$ is an $\bbF^\bbM$ stopping time), so that the right hand side of
\eqref{eq:different-filtration-2} is $\F_t^\bbM$-measurable. Moreover, for $\tau >t$ one has
$V_t^\tau = V_t$ and  $(Z_s^\tau)_{s=0}^t = (Z_s)_{s=0}^t $. Hence  we get for any
bounded measurable functional $g$ on $\mathcal{C}^0 \big([0,T])$ that
\begin{align} \nonumber
 E^Q \big ( h(V_t) \ind{\tau >t} g(  (Z_s)_{s=0}^t) \big ) & =
 E^Q \big ( h(V_t^\tau) \ind{\tau >t} g(  (Z_s^\tau)_{s=0}^t) \big ) \\
& =  E^Q \Big ( E^Q \big ( h(V_t^\tau) \ind{\tau >t} \mid \F_t^{Z^\tau} \vee \F_t^Y \big )
    g((Z_s^\tau)_{s=0}^t) \Big ) \label{eq:different-filtration-3}\,.
\end{align}
Due to the presence of the indicator $\ind{\tau >t}$ in \eqref{eq:different-filtration-3} we may  replace $ g((Z_s^\tau)_{s=0}^t)$
with  $g((Z_s)_{s=0}^t)$ in that equation, so that we obtain \eqref{eq:different-filtration-2} by the definition of conditional expectations. A similar argument shows that
$E^Q \big ( h(V_t) \ind{\tau >t} \mid  \F_t^{\tilde Z} \vee\F_t^Y\big ) =
 E^Q \big ( h(V_t) \ind{\tau >t} \mid \F_t^{Z^\tau} \vee \F_t^Y \big )
$, which  gives the equality
\begin{equation*} 
 E^Q \big ( h(V_t) \ind{\tau >t} \mid \F_t^\bbM \big ) =
 E^Q \big ( h(V_t^\tau) \ind{\tau >t} \mid \F_t^{\tilde Z } \vee \F_t^Y \big )\,,
\end{equation*}
as claimed.
\end{proof}

\begin{proof}[Proof of Proposition~\ref{prop:bounded-domain}]
The process $F_t^{\sigma_N} =  Q(\sigma_N \le t\mid \F_t)$, ${0\le t\le T}$, is an $\bbF$-submartingale. Hence we get, using the first submartingale inequality (see for instance \citeasnoun{bib:karatzas-shreve-88}, Theorem~1.3.8(i))
\begin{equation*}
Q\big( \sup_{s \le T} F_s^{\sigma_N} > \epsilon \big )\le \frac{1}{\epsilon} E\left ( (F_T^{\sigma_N})^+ \right ) = \frac{1}{\epsilon} Q (\sigma_N \le T),
\end{equation*}
which gives Statement 1.

Now we turn to Statement 2. For the purposes of the proof we make the dependence of the stopped asset value process on $N$ explicit and we write $X_t^N := V_t^{\tau \wedge \sigma^N}$. Clearly, on $\{\sigma_N > T\}$ it  holds  that $\ind{\tau^N >t} =  \ind{\tau>t}$, $t\le T$.  Using Proposition~\ref{prop:reduction-to-background} we thus get
\begin{align}\nonumber
& Q\Big( \sup_{t \le T} \Big | { \ind{\tau^N >t} E^Q\big(h(t, V_t^N) \mid \F_t^{\bbM,N} \big)
 - \ind{\tau > t} E^Q \big ( h(t, V_t) \mid \F_t^{\bbM} \big)   }\Big | >\delta  \Big)   \\
&\;  \le   \label{eq:to-be-estimated}
Q\bigg ( \sup_{t \le T} \Big |
\frac{ E^Q\big(h(t, X_t^N) \ind{X_t^N >K}\mid \F_t^{Z^N} \vee \F_t^{D^N}\big)}{Q\big (X_t^N >K \mid \F_t^{Z^N} \vee \F_t^{D^N} \big )} -
\frac{ E^Q\big(h(t, X_t) \ind{X_t >K}\mid \F_t^{Z} \vee \F_t^{D}\big)}{Q\big ( X_t >K \mid \F_t^{Z} \vee \F_t^{D} \big )}
 \Big | > \delta \bigg ) \\
 &\; +  Q \big( \sigma_N > T\big)  \nonumber
\end{align}
Now $ Q \big( \sigma_N > T\big)$ converges to zero for $N \to \infty$ so that we concentrate on \eqref{eq:to-be-estimated}. The difficulty in estimating this probability is the fact that we have to compare conditional expectations with respect to different filtrations. Similarly as in robust filtering, we address this problem using the reference probability approach.  To ease the notation we introduce the abbreviations
$ h_t^N = h(t,X_t^N)\ind{X_t^N >K}$  and $ h_t = h(t,X_t) \ind{X_t > K} \,.$
Moreover, we set
\begin{align*}
 L_t^{1,N}  &= \exp\Big(\int_0^t a(X_s^N )^\top d Z_s  - \frac{1}{2} \int_0^t \abs{a((X_s^N))}^2 \, ds \Big )\; \text{ and }
 L_t^{2,N} = \prod_{t_n \le T} \frac{\varphi(d_n,X_{t_n -}^N)}{\varphi^*(d_n)} \, ,
\end{align*}
and we let $L_t^N:= L_t^{1,N} L_t^{2,N}$. The density martingale
 $L_t  = L_t^1 L_t^2$ is defined in analogously, but with $V^\tau$ instead of $X^N$.
In view of Proposition~\ref{prop:reduction-to-background} and \eqref{eq:kallianpur-striebel-2}   we need to show that for $N \to \infty$,
\begin{equation}\label{eq:key-claim}
\sup_{t \le T} \Big |
\frac{E^{Q_1}\big (h_t^N L_t^N(\cdot ,\omega_2) \big ) }{E^{Q_1}\big (
\ind{X_t^N >K} L_t^N(\cdot ,\omega_2) \big ) }
-\frac{E^{Q_1}\big (h_t L_t(\cdot ,\omega_2) \big ) }{E^{Q_1}\big (
\ind{V_t^\tau >K} L_t(\cdot ,\omega_2) \big ) }
\Big | \overset{Q}{\longrightarrow} 0\,.
\end{equation}
The key tool for this is the following lemma.
\begin{lemma} \label{lemma:bounded-domain} Consider a generic function $f \colon [0,T) \times [K, \infty) \to \R$  with $\abs{f(t,v)} \le c_0 + c_1 v$ and let $f_t^N = f(t,X_t^N)$ and $f_t = f(t, V_t^\tau). $ Then it holds that
$$
\sup_{t \le T} \Big | E^{Q_1}\big (f_t^N L_t^N(\cdot ,\omega_2) \big ) - E^{Q_1}\big (f_t L_t(\cdot ,\omega_2) \big )  \Big | \overset{Q}{\longrightarrow} 0.
$$
\end{lemma}
\begin{proof}[Proof of Lemma~\ref{lemma:bounded-domain}] Fix constants $\epsilon, \delta >0$.
We have to show that for $N$ sufficiently large,
\begin{equation}\label{eq:claim-of-lemma}
Q_2 \Big (\big \{\omega_2 \colon  \sup_{t \le T} \Big | E^{Q_1}\big (f_t^N L_t^N(\cdot ,\omega_2) \big ) - E^{Q_1}\big (f_t L_t(\cdot ,\omega_2) \big )  \Big | >\delta \big \}\Big) < \epsilon
\end{equation}
We first show that we may assume without loss of generality that $L_t^{2,N}$ and $L_t^2$ are bounded. To this we choose some constant $C$ such that $Q_2(B) > 1- \frac{\epsilon}{2}$ where
$$
B : =\Big \{ \omega_2 \colon \sup_{x \in [K, \infty)} \prod_{t_n \le T} \frac{\varphi(d_n (\omega_2), x)}{\varphi^*(d_n(\omega_2))} \le C\Big\}
$$
This is possible since $\varphi^* $ is bounded away from zero on every compact subset of $(0, \infty)$ and since for fixed $d_n$ the mapping $x \mapsto \varphi(d_n,x)$ is bounded by \eqref{eq:bound-on-density}.  By definition it holds on $B$ that $L_t^{2,N} = L_t^{2,N}\wedge C$ and $L_t^2 = L_t^{2}\wedge C$. Moreover,  \eqref{eq:claim-of-lemma} is no larger than
\begin{align*}
& Q_2 \bigg ( \Big \{ \sup_{t \le T} \Big | E^{Q_1}\big (f_t^N L_t^{N,1 }(L_t^{N,2} \wedge C) (\cdot ,\omega_2) \big ) - E^{Q_1}\big (f_t L_t^1 (L_t^2 \wedge C) (\cdot ,\omega_2) \big )  \Big | >\delta \Big \} \cap B \bigg) + Q(B^c)\\
& \quad \le Q_2 \Big (  \sup_{t \le T} \Big | E^{Q_1}\big (f_t^N L_t^{N,1 }(L_t^{N,2} \wedge C) (\cdot ,\omega_2) \big ) - E^{Q_1}\big (f_t L_t^1 (L_t^2 \wedge C) (\cdot ,\omega_2) \big )  \Big | >\delta  \Big) + \frac{\epsilon}{2}
\end{align*}
Hence we assume from now on that $L_t^{2,N}$ and $L_t^2$ are bounded by $C$.

We continue with some useful estimates on $L_t$. Doob's maximal inequality gives
\begin{align} \nonumber
E^{Q^*} \bigg (\sup_{0 \le t \le T} (L_t)^2 \bigg ) & \le C^2 E^{Q^*}\Big ( \sup_{0 \le t \le T} (L_t^1)^2 \Big) \le 4 C^2 E^{Q^*} \Big(  (L_T^1)^2 \Big ) = 4 C^2  E^{Q^*} \bigg(  \exp\Big (\int_0^T \norm{a(V_s^\tau)}^2 ds \Big ) \bigg )\\
& \le  4 C^2  \exp \Big ( T \sup_{x \ge K} \norm{a(x)}^2 \Big ) \label{eq:ucp-proof-3}
\end{align}
Similarly we get for  $\bar V_t = \tilde V_0 \exp\big( (r- \frac{1}{2}\sigma^2) t + \sigma B_t\big )$ (the cum-dividend asset value) that
\begin{align} \nonumber
E^{Q^*} \bigg (\sup_{0 \le t \le T} (\bar V_t L_t)^2 \bigg) &
\le C^2 e^{2 r T} E^Q(\tilde V_0^2) E^{Q^*}\bigg ( \sup_{0 \le t \le T} \mathcal{E}\big ( a(V_s^\tau)^\top d Z_s + \sigma dB_s \big)_t \bigg)\\
&\le  4 C^2  e^{2 r T} E^{Q^*}(\tilde V_0^2) \exp \Big ( T ( \sigma^2 + \sup_{x \ge K} \norm{a(x)}^2) \Big ) \label{eq:ucp-proof-4}
\end{align}
Of course, similar estimates hold for $E^{Q^*} \left (\sup_{0 \le t \le T} (L_t^N)^2 \right)$ and for $E^{Q^*} \left (\sup_{0 \le t \le T} (\bar V_t L_t^N)^2 \right)$.
Since on the set $\{ t < \sigma_N \}$, $f_t^N = f_t$ and $L_t^N = L_t$,
it holds that
\begin{align*}
 &\Big \vert E^{Q_1}\big( f_t^N L_t^N(\cdot, \omega_2)\big)  - E^{Q_1} \big ( f_t L_t(\cdot, \omega_2) \big)   \Big \vert
 \le E^{Q_1}\Big(\ind{\sigma_N \le t}  \abs{f_t^N} L_t^N \big(\cdot, \omega_2)\Big )+
E^{Q_1}\Big(\ind{\sigma_N \le t}\abs {f_t} L_t(\cdot, \omega_2) \Big) .
\end{align*}
Hence we get
\begin{align}\nonumber
&\sup_{0 \le t \le T} \Big \vert E^{Q_1}\big( f_t^N L_t^N(\cdot, \omega_2)\big)  - E^{Q_1} \big ( f_t L_t(\cdot, \omega_2) \big)   \Big \vert \\
&\quad \le \sup_{0 \le t \le T} E^{Q_1}\big(\ind{\sigma_N \le t}  \abs{f_t^N} L_t^N \big(\cdot, \omega_2)\big )+
\sup_{0 \le t \le T}  E^{Q_1}\big(\ind{\sigma_N \le t}\abs {f_t} L_t(\cdot, \omega_2) \big) . \label{eq:ucp-proof-5}
\end{align}
Now note that by assumption $\abs{f_t^N} \le c_0 + c_1 X_t^N \le c_0 + c_1 \bar V_t $. Hence \eqref{eq:ucp-proof-5} can be estimated by
\begin{equation}\label{eq:ucp-proof-6}
\sup_{0 \le t \le T} E^{Q_1}\Big(\ind{\sigma_N \le t}  (c_0 + c_1 \bar V_t) L_t^N (\cdot, \omega_2)\Big ) + \sup_{0 \le t \le T} E^{Q_1}\Big(\ind{\sigma_N \le t}  (c_0 + c_1 \bar V_t) L_t (\cdot, \omega_2)\Big )
\end{equation}
In order to complete the proof of the lemma we finally show that the expression
$$ E^{Q_2} \Big (\sup_{0 \le t \le T} E^{Q_1}\big(\ind{\sigma_N \le t}  (c_0 + c_1 \bar V_t) L_t^N \big(\cdot, \omega_2)\big ) \Big ) =
 E^{Q^*} \Big (\sup_{0 \le t \le T} \ind{\sigma_N \le t}  (c_0 + c_1 \bar V_t) L_t^N  \Big )
$$
converges to zero for $N \to \infty$, that is $E^{Q_1}\big(\ind{\sigma_N \le t}  (c_0 + c_1 \bar V_t) L_t^N (\cdot, \omega_2)\big )$ converges to zero in $L^1(\Omega_2, \F_2, Q_2)$ and hence also in probability. Now note that our previous estimates \eqref{eq:ucp-proof-3} and \eqref{eq:ucp-proof-4} imply that the random variables $Y^N : =  \sup_{0 \le t \le T} \ind{\sigma_N \le t} (c_0 + c_1 \bar V_t)$ are uniformly bounded in $L^2(\Omega, \F, Q^*)$ and hence uniformly integrable, so that the claim follows from the Lebesgue theorem. The same argument obviously applies to the second term in
\eqref{eq:ucp-proof-6} which proves the Lemma.
\end{proof}
Finally we return to the proof of \eqref{eq:key-claim} and hence of Proposition~\ref{prop:bounded-domain}. Fix constants $\epsilon, \delta >0$ and choose $M>0$ in such a way that for $A_1,A_2$ with
$$ A_1 := \big \{\omega_2 \colon \inf_{t \le T} E^{Q^1} \big( L_t(\cdot, \omega_2)\big) > {2}/{M} \big \}\;  \text { and } A_2 := \big \{\omega_2 \colon \sup_{t \le T} E^{Q_1} \big ( h_t L_t(\cdot, \omega_2)\big ) < M- \delta \big \}
$$
it holds that $Q_2(A_1^c) < \frac{\varepsilon}{4}$ and $Q_2(A_2^c) < \frac{\varepsilon}{4}$. Choose finally $N_0$ large enough so that for $N > N_0$  it holds  $Q_2(A_3) > 1- \frac{\varepsilon}{4}$ and $Q_2(A_4) > 1- \frac{\varepsilon}{4}$; here
\begin{align*}
A_3&:= \big\{\omega_2 \colon  \sup_{t \le T} \Big | E^{Q_1}\big (h_t^N L_t^N(\cdot ,\omega_2) \big ) - E^{Q_1}\big (h_t L_t(\cdot ,\omega_2) \big )  \big | < \frac{\delta}{2M} \big\}  \\
A_4&:= \big \{\omega_2 \colon  \sup_{t \le T} \Big | E^{Q_1}\big (L_t^N \ind{X_t^N >K}(\cdot ,\omega_2) \big ) - E^{Q_1}\big ( L_t\ind{V_t^\tau >K}(\cdot ,\omega_2) \big )  \big | < \frac{\delta}{2M} \big \}\,;
\end{align*}
this is possible by Lemma~\ref{lemma:bounded-domain}. Let $A = A_1 \cap A_2 \cap A_3 \cap A_4$ and note that $Q(A)  >1 -\epsilon$. By definition of $A$ we have for $N > N_0$ and $\omega_2 \in A$ the estimates
$ \sup_{t < T} E^{Q_1}\big( h_t^N L_t^N (\cdot, \omega_2) \big)< {M}$ {and } $\inf_{t < T} E^{Q_1}\big(  L_t^N (\cdot, \omega_2) \big)>\frac{1}{M}$. Now  the mean value theorem from standard calculus gives for $x,\tilde x$ and $y,\tilde y \in \R$ with $\abs{x}, \abs {\tilde x}<M$ and $\abs{y}, \abs{\tilde y } > \frac{1}{M} $ the estimate
$\big \vert \frac{\tilde x}{\tilde y} - \frac{x}{y} \big \vert
\le M\left ( \abs{\tilde x -x} +\abs{\tilde y -y} \right).$
Applying this estimate we get for $\omega_2 \in A$ that
\begin{align*}
 & \sup_{t \le T} \Big | \frac{E^{Q_1}\big (h_t^N L_t^N(\cdot ,\omega_2) \big ) }{
E^{Q_1}\big (\ind{X_t^N >K} L_t^N(\cdot ,\omega_2) \big ) }
-\frac{E^{Q_1}\big (h_t L_t(\cdot ,\omega_2) \big ) }{
E^{Q_1}\big (\ind{V_t^\tau >K} L_t(\cdot ,\omega_2) \big ) } \Big |
  \le M \sup_{t \le T} \Big \{ \\
  & \qquad  \abs{E^{Q_1}\big (h_t^N L_t^N(\cdot ,\omega_2) \big ) - E^{Q_1}\big (h_t L_t(\cdot ,\omega_2) \big )}
    + \abs{E^{Q_1}\big (L_t^N (\cdot ,\omega_2) \ind{X_t^N >K} \big ) - E^{Q_1}\big (L_t\ind{V_t^\tau >K}(\cdot ,\omega_2) \big )}\Big\}
\\ & \quad
\le M \Big(\frac{\delta}{2M} + \frac{\delta}{2M} \Big ) = \delta
\end{align*}
\end{proof} 

\begin{proof}[Proof of Lemma~\ref{lemma:density-for-dividends}]

In order to show that
$E^{Q^*}(L_T) =1$ we use induction over the dividend dates $t_n$. For $t<t_1$, $L_t=L_t^1$ and $ E^{Q^*}\big( L_t^1 \big) =1$ by the Novikov criterion (recall that $a$ is bounded by assumption). Suppose now that the claim holds for $t<t_n$. We get that
$$  E^{Q^*} \big (L_{t_n}\big)  =  E^{Q^*} \Big (L_{t_n -}\, \frac{\varphi(d_n,X_{t_n-})} {\varphi^*(d_n)}  \Big) = E^{Q^*} \Big (L_{t_n -}\, E^{Q^*} \Big ( \frac{\varphi(d_n,X_{t_n-})} {\varphi^*(d_n)}  \mid \F_{t_n -} \Big)\Big ).
$$
Moreover,
$$ E^{Q^*} \Big ( \frac{\varphi(d_n,X_{t_n-})} {\varphi^*(d_n)}  \mid \F_{t_n -} \Big) = \int_{\R^+} \frac{\varphi(y,X_{t_n-})} {\varphi^*(y) } \varphi^*(y) dy = \int_{\R^+} {\varphi(y,X_{t_n-})} \, dy =1\,,
$$
so that $E^{Q^*}(L_{t_n}=1)$ as well.

In order to show that $\mu^D(dy,dt)$ has $Q$-compensator $\gamma^D(dy,dt)$ we use the general Girsanov theorem (see for instance \citeasnoun{bib:protter-05}, Theorem~3.40): a process $M$ such that $\langle M,L \rangle $ exists for $Q^*$  is a $Q^*$-local martingale if and only if
$\widetilde{M}_t= M_t - \int_0^t \frac{1}{L_s-} d \langle L,M\rangle_s $ is a $Q$-local martingale.

Consider now some bounded predictable function $\beta\colon [0.T]\times\R^+ \to \R$ and define  the $Q^*$-local martingale  $M_t = \int_0^t \int_{\R^+} \beta(s,y) \, (\mu^D-\gamma^{D,*})(dy,ds)$. As $M$ is of finite variation, we get that
$$[M,L]_t =  \sum_{s \le t} \Delta M_s \Delta L_s = \int_0^t \int_{\R^+} L_{s-} \big(  \frac{\varphi(y,X_{s-})}{\varphi^*(y)} - 1 \big ) \beta(s,y) \mu^D(dy,ds).$$
Hence we get that
\begin{align} \label{eq:QUADVAR}
\langle M,L\rangle_t &= \int_0^t \int_{\R^+} L_{s-} \big(  \frac{\varphi(y,X_{s-})}{\varphi^*(y)} - 1 \big ) \beta(s,y) \gamma^{D,*}(dy,ds)\,,
\end{align}
provided that
\begin{equation} \label{eq:ESTIMATE-FOR-QUADVAR}
 E^{Q^*} \Big ( \int_0^t \int_{\R^+} L_{s-} \big | \frac{\varphi(y,X_{s-})}{\varphi^*(y)} - 1 \big | |\beta(s,y)| \gamma^{D,*}(ds,dy) \Big )  < \infty \, .
\end{equation}
Recall that $\gamma^{D,*}(dy,dt) = \sum_{n=1}^\infty \varphi^*(y) dy\, \delta_{\{t_n\}} (dt)$. Hence
$$\int_0^t \int_{\R^+} L_{s-} \big | \frac{\varphi(y,X_{s -})}{\varphi^*(y)} - 1 \big | \gamma^{D,*}(ds,dy) =
\sum_{t_n \le t} L_{t_n -} \int_{\R^+} \big | \varphi(y,X_{t_n -}) - \varphi^*(y) \big | dy \le 2 \sum_{t_n \le t} L_{t_n -}.$$
Since $|\beta|$ is bounded by some constant $C $ and since $E^{Q^*}(L_t)  = 1 $ for all $t$, the lhs of \eqref{eq:ESTIMATE-FOR-QUADVAR} is bounded by $2 C \sup \{n \in \N \colon t_n \le t \}$. Moreover,  we get from \eqref{eq:QUADVAR} that
$$
\langle M,L\rangle_t = \sum_{t_n \le t} L_{t_n -} \int_{\R^+}  \beta(t_n,y)  \big ( \varphi(y,X_{t_n-}) - \varphi^*(y) \big ) dy
$$
This gives
\begin{align*}
\widetilde{M}_t :=  M_t - \int_0^t \frac{1}{L_s-} d \langle L,M\rangle_s = \int_0^t \int_{\R^+} \beta(s,y) (\mu^D-\gamma^D)(dy,ds)\,.
\end{align*}
Now $\widetilde{M}$ is a local martingale by the general Girsanov theorem, which shows that
 $\gamma^D(dy,dt)$ is in fact the $Q$-compensator of $\mu^D$. The other claims are clear.
\end{proof}

\begin{proof}[Proof of Theorem~\ref{prop:dividend-information}] We proceed via induction over the dividend dates. For $t \in [0, t_1)$ there is no dividend information  and the claim follows from Theorem~\ref{thm:SPDE-for-u}. Suppose now
that the claim of the theorem holds for $t \in [0, t_n)$. First we show that
$$
u(t_n) = S_{\kappa d_n} ( \tilde u (t_n)) = u(t_{n-}, x + \kappa d_n)  \frac{\varphi(d_n,x + \kappa d_n)}{\varphi^*(d_n)}
$$
belongs to  $H^1_0 (S^X) \cap H^2(S^X)$. Clearly, under Assumption~\ref{ass:assets-and-dividends}2, for a  given $d_n >0$, the mapping $x \mapsto   {\varphi(d_n,x)}\big/{\varphi^*(d_n)}$ is smooth, nonnegative  and bounded. Hence with $u(t_{n-})$  also $\tilde u (t_n))$ belongs to $H^1_0 (S^X) \cap H^2(S^X)$. If $\kappa =1$
it remains to show that $S_{d_n} ( \tilde u (t_n))$ is an element of $H^1_0 (S^X) \cap H^2(S^X)$.  Smoothness is clear, the only thing that needs to be verified is the boundary condition $ \tilde u(t_n, K+ d_n) = 0\,. $
To this note that $\varphi(d_n, z)= 0$ for $z \le d_n + K$ (see equation~\eqref{eq:dividend-density}). This implies that $\tilde u (t_n,  K +  d_n) =0$ as required. Existence and uniqueness of a solution of \eqref{eq:du-t-with-jumps} on $[t_n, t_{n+1}) $ follows then immediately from Theorem~\ref{thm:SPDE-for-u}.

Next we turn to the second claim. Using the induction hypothesis, the definition of $\tilde u$, the fact that $\varphi(d_n,K) = \varphi^*(d_n)$ and the dynamics of $\nu_K(t) $ and $\nu_N(t)$ we get that
\begin{align*}
 \Sigma_{t_n} f  &= E^{Q^*} \Big( L_{t_n-} f(X_{t_n}) \frac{\varphi(d_n,X_{t_n -})}{\varphi^*(d_n)}\Big) \\
   &= \big ( u(t_{n }- ), f(\cdot - \kappa d_n) \frac{\varphi(d_n,\cdot)}{\varphi^*(d_n)}\big )_{S^X} + \nu_K (t_{n}-) f(K) \frac{\varphi(d_n,K)}{\varphi^*(d_n)}+ \nu_N(t_{n}-) f(N) \frac{\varphi(d_n,N)}{\varphi^*(d_n)}\\
   &= \big ( \tilde u(t_{n }- ), f(\cdot - \kappa d_n)\big )_{S^X} + \nu_K(t_n) f(K) + \nu_N(t_n) f(N).
\end{align*}
Now we  have
\begin{align*}\big ( \tilde u(t_{n } ), f(\cdot - \kappa d_n)\big )_{S^X} &= \int_K^N f(x - \kappa d_n) \tilde u(t_{n }, x) dx = \int_{K- \kappa d_n}^{N-\kappa d_n} f(y)  \tilde u(t_{n }, y + \kappa d_n)\, dy\\
&= \int_{K}^{N} f(y)  \tilde u(t_{n }, y + \kappa d_n)\, dy,
\end{align*}
as the integrand is zero  on    $[K- \kappa d_n, K] \cup [N - \kappa d_n, N]$. Hence we get that $$ \big ( \tilde u(t_{n } ), f(\cdot - \kappa d_n)\big )_{S^X} = \big (  u(t_{n } ), f \big )_{S^X}$$ and thus the relation
$\Sigma_{t_n} f = \big ( u(t_n), f\big )_{S^X} + \nu_K (t_n) f(K) + \nu_N(t_n) f(N)$ as claimed.
\end{proof}

\begin{proof}[Proof of Lemma~\ref{lem:d-pit-f}] We start with the case without dividends and we assume that   $f$ is time-independent. Recall from Corollary~\ref{corr:cond-density} that
 $\pi_t f = \frac{1}{C(t)} \left ( (u(t),f  )_{S^X} + \nu_N(t) f(N) \right )$ with $C(t)
 = (u(t),1  )_{S^X} + \nu_N(t) $. Using \eqref{eq:d(u,v)-t-2} and \eqref{eq:dynamics-nuN} we get
 \begin{equation} \label{eq:dCt-1}
d C(t) = \Big(   \big ( \mathcal{L}^* u(t), 1 \big )_{S^X} - \frac{1}{2} \sigma^2 N^2
\frac{du}{dx}(t,N) \Big)\, dt + \left(\big( u(t), a^\top\big)_{S^X} + \nu_N(t) a^\top(N)
\right) dZ_t.
 \end{equation}
Partial integration shows that the drift term in \eqref{eq:dCt-1} equals $ - \frac{1}{2}
\sigma^2 K^2 \frac{du}{dx}(t,K)$. Hence,
\begin{align*}
d \frac{1}{C(t)} &= \frac{1}{2 C(t)^2}\sigma^2 K^2 \frac{du}{dx}(t,K)\, dt -
  \frac{1}{C(t)^2} \left(\big( u(t), a^\top\big)_{S^X} + \nu_N(t) a^\top(N) \right) \,dZ_t\\
  & + \frac{1}{C(t)^3} \sum_{j=1}^{l} \big (\big( u(t), a_j\big)_{S^X} + \nu_N(t) a_j \big)^2
  \, dt\,.
\end{align*}
Similarly, we obtain that
\begin{equation}
\begin{split}
d \Big( \big( u(t),f  \big)_{S^X} + \nu_N (t) f(N)\Big) &= \Big ( \big ( u(t), \mathcal{L}_X f\big )_{S^X} - \frac{1}{2} \sigma^2 K^2 \frac{du}{dx}(t,K) f(K)  \Big) \, dt \\ & + \Big
(\big( u(t), a^\top f \big)_{S^X} + \nu_N(t) a^\top(N) f(N) \Big)\, dZ_t.
\end{split}
\end{equation}
Hence we get, using the It\^o product formula and the fact that $\pi(t,v) = u(t,v)/C(t)$
\begin{align*}
d \pi_t f =& \frac{1}{C(t)} d\left ( \big(u(t),f \big )_{S^X} + \nu_N(t) f(N)\right ) + \left ( \big(u(t),f \big )_{S^X} + \nu_N(t) f(N)\right ) d\frac{1}{C(t)} \\ & \hspace{3.15cm}+ d \, \big [
\frac{1}{C},\big(u,f\big )_{S^X} + \nu_N f(N) \big ]_t \\
= & \Big( \big ( \pi(t), \mathcal{L}_X f\big )_{S^X} - \frac{1}{2} \sigma^2 K^2
 \frac{d\pi}{dx}(t,K) f(K)  \Big) \, dt + \pi_t(a^\top f) \, dZ_t \\
 +& \Big ( \frac{1}{2} \sigma^2 K^2 \frac{d\pi}{dx}(t,K) \,\pi_t f + \sum_{j=1}^{l} (\pi_t
a_j)^2 \pi_t f \Big ) \, dt\ - (\pi_t f ) (\pi_t a^\top )\, dZ_t - \Big (\sum_{j=1}^{l}
(\pi_t a_j) (\pi_t a_j f)\Big ) \, dt\,.
\end{align*}
Rearranging terms and using that $\pi_t (\mathcal{L}_X f) = (\pi(t), \mathcal{L}_X f )_{S^X}$ gives the claim of the lemma for the case of time-independent $f$ and no dividends. For time-dependent $f$ we have $d \pi_t f
 =  \pi_t (\frac{df}{dt}(t,\cdot))  \, dt + d \pi_t f(t,\cdot)$ so that we
obtain the additional  term  $\pi_t (\frac{df}{dt}(t,\cdot))$ in the drift of the
dynamics of $\pi_t f$.

Finally, we consider  the case with dividend
payments. Between dividend dates the dynamics of $\widehat{f}_t$ can be derived by
similar arguments as before. To derive the jump of $\pi_t f$ at $t_n$ recall that
$\pi_{t_n} f = \Sigma_{t_n} (f 1_{(K, \infty)}) \big / \Sigma_{t_n} (1_{(K, \infty)})$.
Moreover we have, using the form of $L_{t_n}$,

$$ \Sigma_{t_n} (f 1_{(K, \infty)}) = \frac{1}{\varphi^*(d_n)} \Sigma_{t_n -} \Big( f(\cdot - \kappa d_n) 1_{(K, \infty)} (\cdot - \kappa d_n) \varphi(d_n, \cdot) \Big)
$$
and similarly for  $\Sigma_{t_n}  (1_{(K, \infty)})$. Now, by Assumption~\ref{ass:assets-and-dividends}2 we know that $d_n < (X_{t_n}- K)^+$ $Q$ a.s, so that $1_{(K, \infty)} (X_{t_n-} - \kappa d_n) = 1_{(K, \infty)} (X_{t_n-})$.  Hence we get
\begin{align*}
\pi_{t_n} f   & = \frac{   \Sigma_{t_n} (f 1_{(K, \infty)}) \big/ \Sigma_{t_n - } (1_{(K, \infty)})}{ \Sigma_{t_n} (1_{(K, \infty)}) \big/ \Sigma_{t_n - } (1_{(K, \infty)})}
  =
\frac{   \Sigma_{t_n - } \big ( f(\cdot - \kappa d_n) 1_{(K, \infty)}  \varphi(d_n, \cdot) \big) \big/ \Sigma_{t_n - } ( 1_{(K, \infty)})}{
\Sigma_{t_n} \big (  1_{(K, \infty)} \varphi(d_n, \cdot) \big) \big/ \Sigma_{t_n - } (1_{(K, \infty)})}
\\ & =
\frac{\pi_{t_n -} \big( f(\cdot - \kappa d_n) \varphi(d_n, \cdot) \big)}{
\pi_{t_n -} \big(  \varphi(d_n, \cdot) \big)},
\end{align*}
which gives the form of the integral with respect to $\mu^D(dy,ds)$ in \eqref{eq:filter-equation}.
\end{proof} 

}


\end{document}